\documentclass[a4paper, 12pt]{article}

\usepackage{fullpage}
\usepackage{amsthm}
\usepackage{amsmath}
\usepackage{amssymb}

\usepackage{algorithm}
\usepackage{algorithmic}

\theoremstyle{plain}
\newtheorem{theorem}{Theorem}
\newtheorem{lemma}[theorem]{Lemma}
\newtheorem{corollary}			[theorem]	{Corollary}
\newtheorem{proposition}		[theorem]	{Proposition}
\theoremstyle{definition}
\newtheorem{definition}[theorem]{Definition}
\newtheorem{remark}[theorem]{Remark}

\DeclareMathOperator{\ISOTYPE}{ISOTYPE}

	\usepackage[utf8]{inputenc}
	\usepackage[T1]{fontenc}
	\usepackage[ngerman,american]{babel}
	\usepackage[usenames,dvipsnames,svgnames,table]{xcolor}

	\usepackage{amsmath,enumitem,leftidx,framed,graphicx,clrscode,parskip,changepage,tikz,float}

	\usepackage{hyperref}
	\usepackage{graphicx}
		\graphicspath{ {./pictures/} }
		
		\hypersetup{colorlinks, linkcolor=teal, citecolor=teal, urlcolor=teal}
	\usepackage{amssymb}

		\makeatletter
			\def\thm@space@setup{%
			  \thm@preskip=\parskip \thm@postskip=0pt
			}
		\makeatother

\renewenvironment{proof}[1][\proofname]{{\gdef\QED{\qed\gdef\QED{}}\noindent\bfseries #1:} }{\QED}

\usepackage{xspace}

\usepackage{mathtools}
\usepackage{subfig}
\usepackage{enumitem}

\usetikzlibrary{decorations.markings}
\usetikzlibrary{shapes.geometric}
\usetikzlibrary{calc}

\pgfdeclarelayer{edgelayer}
\pgfdeclarelayer{nodelayer}
\pgfsetlayers{edgelayer,nodelayer,main}
\tikzstyle{rn}=[circle,fill=white,draw=black,line width=0.6 pt]
\tikzstyle{standardedge}=[-,draw=black,line width=2.000]
\definecolor{lightlightgray}{gray}{0.95}

\usepackage{etex,etoolbox}

\usepackage[normalem]{ulem}

\newcommand{\uend}{\hfill$\lrcorner$}

\newcounter{claimcounter}
\newenvironment{claim}[1][]{
  \renewcommand{\proof}{\smallskip\par\noindent\textit{Proof of the claim. }}
  \medskip\par\noindent%
  \ifthenelse{\equal{#1}{}}{%
    \setcounter{claimcounter}{0}\refstepcounter{claimcounter}\textit{Claim~\arabic{claimcounter}.}
  }{%
    \ifthenelse{\equal{#1}{resume}}{%
      \refstepcounter{claimcounter}\textit{Claim~\arabic{claimcounter}.}
    }{%
      \textit{Claim~#1.}
    }
  }
}{
  \par\medskip
}

\newlist{proofcases}{enumerate}{2}
\SetEnumitemKey{noindent}{leftmargin=0pt,itemindent=0em,labelwidth=\itemindent-\labelsep,align=left}
\SetEnumitemKey{nospace}{topsep=5pt,itemsep=5pt}
\SetEnumitemValue{before}{linebreak}{\needspace{1\baselineskip}\leavevmode}
\newcommand\outerCase{\arabic} 
\newcommand\innerCase{\arabic} 
\setlist[proofcases]{noindent,nospace}
\setlist[proofcases,1]{label=\textit{(Case \outerCase*\protect\thiscase) \ },ref=\outerCase*}
\setlist[proofcases,2]{label=\textit{(Case \outerCase{proofcasesi}.\innerCase*\protect\thiscase) \ },ref=\outerCase{proofcasesi}.\innerCase*}

\newcommand{\proofcase}[1][]{%
  \if\relax\detokenize{#1}\relax
    \def\thiscase{.}%
  \else
    \def\thiscase{: \ #1.}%
  \fi
  \item}

\begin{document}

\title{The Weisfeiler-Leman Dimension of Planar Graphs is at most~3}

\author{Sandra Kiefer\\
RWTH Aachen University\\
kiefer@informatik.rwth-aachen.de
\and
Ilia Ponomarenko \\
Petersburg Department of V.~A.~Steklov 
Institute of Mathematics\\
inp@pdmi.ras.ru
\and Pascal Schweitzer\\
RWTH Aachen University \\
schweitzer@informatik.rwth-aachen.de}

\maketitle

\insert\footins{\noindent \footnotesize This is an extended version of the paper with the same title published in the Proceedings of the 32nd Annual ACM/IEEE Symposium on Logic in Computer Science \cite{8005107}.}

\begin{abstract}
We prove that the Weisfeiler-Leman (WL) dimension of the class of all finite planar graphs is at most~$3$. In particular, every finite planar graph is definable in first-order logic with counting using at most~$4$ variables. The previously best known upper bounds for the dimension and number of variables were~$14$ and~$15$, respectively.

First we show that, for dimension~3 and higher, the WL-algorithm correctly tests isomorphism of graphs in a minor-closed class whenever it determines the orbits of the automorphism group of any arc-colored~$3$-connected graph belonging to this class.

Then we prove that, apart from several exceptional graphs (which have WL-dimension at most~2), the individualization of two correctly chosen vertices of a colored~3-connected planar graph followed by the~1-dimensional WL-algorithm produces the discrete vertex partition. This implies that the~3-dimensional WL-algorithm determines the orbits of a colored~$3$-connected planar graph.

As a byproduct of the proof, we get a classification of the~3-connected planar graphs with fixing number~$3$.
\end{abstract}
\bigskip


\section{Introduction}

The Weisfeiler-Leman algorithm (WL-algorithm) is a fundamental algorithm used as a subroutine in graph isomorphism testing. More precisely, it constitutes a family of algorithms. For every positive integer~$k$ there is a~$k$-dimensional version of the algorithm that colors all~$k$-tuples of vertices in two given undirected input graphs and iteratively refines the color classes based on information of previously obtained colors. 

The algorithm has surprisingly strong links to notions that seem unrelated at first sight. For example, there is a precise correspondence to Sherali-Adams relaxations of certain linear programs
\cite{DBLP:journals/siamcomp/AtseriasM13,DBLP:conf/csl/GroheO12}, there are duplicator-spoiler games capturing the same information as the algorithm~\cite{DBLP:journals/combinatorica/CaiFI92}, it is related to separability of coherent configurations~\cite{DBLP:journals/combinatorics/EvdokimovP00}, and there is a close correspondence between the algorithm and first-order logic with counting ($C^k$). More precisely, for two graphs~$G$ and~$G'$, if the integer~$k$ is the smallest dimension such that the~$k$-dimensional WL-algorithm distinguishes two graphs, then~$k+1$ is the smallest number of variables of a sentence in first-order logic with counting distinguishing the two graphs~\cite{DBLP:journals/combinatorica/CaiFI92}. 

Exploiting these correspondences, the seminal construction of Cai, F\"{u}rer and Immerman~\cite{DBLP:journals/combinatorica/CaiFI92} shows that there are examples of pairs of graphs on~$n$ vertices for which a dimension of~$\Omega(n)$ is required for the WL-algorithm to distinguish the two graphs.

However, for various graph classes a bounded dimension suffices to distinguish every two non-isomorphic graphs from each other. While typically not very practical due to large memory consumption, this yields a polynomial-time algorithm to test isomorphism of graphs from such a class.
 In a tour de force, Grohe shows that for all graph classes with excluded minors a bounded dimension of the WL-algorithm suffices to decide graph isomorphism~\cite{DBLP:journals/jacm/Grohe12,grohe_2017}.
 
An ingredient in Grohe's proof deals with the special case of the class of planar graphs, for which a bound on the necessary dimension of the WL-algorithm had been proven earlier by Grohe separately~\cite{DBLP:conf/lics/Grohe98}. Thus, there is a~$k$ such that the~$k$-dimensional WL-algorithm distinguishes every two non-isomorphic planar graphs from each other. In his Master's thesis~\cite{Redies} (see~\cite[Subsection~18.4.4]{grohe_2017}), Redies analyses Grohe's proof showing that this~$k$ can be chosen to be~14. For~3-connected planar graphs it was also shown earlier by Verbitsky~\cite{DBLP:conf/stacs/Verbitsky07} that one can additionally require the quantifier rank to be logarithmic. Feeling that this is far from optimal, Grohe asked in his book~\cite[Subsection~18.4.4]{grohe_2017} and also at the~2015 Dagstuhl meeting on the graph isomorphism problem~\cite{DBLP:journals/dagstuhl-reports/BabaiDST15} for a tight~$k$.
In this paper we show that~$k=3$ is sufficient.

We say that a graph is \emph{identified} by the~$k$-dimensional WL-algorithm if the algorithm distinguishes the graph from all other graphs.
Following Grohe~\cite[Definition~18.4.3]{grohe_2017}, we say that a graph class~$\mathcal{C}$ has \emph{WL-dimension}~$k$ if~$k$ is the smallest integer such that all graphs in~$\mathcal{C}$ are identified by the~$k$-dimensional WL-algorithm. With this terminology our main theorem reads as follows.

\begin{theorem}\label{thm:main}
 The Weisfeiler-Leman dimension of the class of planar graphs is at most~3.\\
(Equivalently, every planar graph is definable by a sentence in first-order logic with counting that uses only~4 variables.)
 \end{theorem}
 
Our proof is separated into two parts. The first part (Sections~\ref{sec:decomps}--\ref{sec:reduct:to:3:con}) constitutes a reduction from general graphs to~$3$-connected graphs, while the second part (Section~\ref{sec:3:con}) handles~3-connected planar graphs.

In the first part, which does not only concern planar graphs, we start by showing that for a hereditary graph class~$\mathcal{G}$, 
for~$k \geq 2$, if the~$k$-dimensional WL-algorithm distinguishes every two vertex-colored non-isomorphic~2-connected graphs from each other, then it distinguishes every two non-isomorphic graphs in~$\mathcal{G}$
(Theorem~\ref{thm:reduction2}).

While it is tempting to believe that when requiring~$k\geq 3$ a similar statement can be made about~3-connected graphs, we need the additional assumption that~$\mathcal{G}$ is minor-closed. In fact, our proof also needs a more technical requirement that the WL-algorithm correctly determines the vertex orbits.

We can then argue that if~$\mathcal{G}$ is additionally minor-closed, for~$k \geq 3$, if the~$k$-dimensional WL-algorithm correctly determines orbits on all arc-colored~$3$-connected graphs in~$\mathcal{G}$, then it distinguishes every two non-isomorphic graphs in~$\mathcal{G}$ from each other (Theorem~\ref{thm:reduction3}).

To prove these two reductions, we employ several structural observations on decomposition trees. These allow us to cut off isomorphism-invariantly the leaves of the decomposition trees of~2- and~3-connected components, respectively. This can be done implicitly without having to explicitly construct the corresponding decomposition trees (see Section~\ref{sec:decomps}).

In the second part we show that the~3-dimensional WL-algorithm identifies all (arc-colored)~3-connected planar graphs. More precisely, we argue that orbits are determined, as required by our reduction. 
In fact, we show a stronger statement in that we do not need the full power of the~3-dimensional WL-algorithm. Using Tutte's Spring Embedding Theorem \cite{MR0158387}, we argue that if in an arc-colored~3-connected planar graph there are three vertices each with a unique color (so-called singletons) that share a common face then applying the~1-dimensional WL-algorithm (usually called color refinement) yields a coloring of the graph in which all vertices are singletons. Since this would only give us a bound of~$k=4$, we show then that in most~3-connected planar graphs it suffices to individualize~2 vertices to get the same result.
Our proof actually characterizes the exceptions, the graphs in which we need to individualize~3 vertices. We can handle these graphs separately to finish our proof. 

The fixing number of a graph~$G$ is the minimum size of a set of vertices~$S$ such that the only automorphism that fixes~$S$ pointwise is the identity. It follows from generally known facts that the fixing number of a~3-connected planar graph is at most~3. Our proof however shows that the only~3-connected planar graphs with fixing number~3 are those depicted in Figure~\ref{fig:3:con:plan:gra:fix:3} (see also~Corollary~\ref{cor:fix:3:char}). The properties of these graphs are summarized in Table~\ref{t1}.

\newcommand{\innerscalar}{0.35}
\newcommand{\outerscalar}{1.1}
\newcommand{\middlescalar}{0.7}
\newcommand{\outvertexh}{1.7}
\newcommand{\outvertexfirstbez}{1.5}
\tikzstyle{normalvertex}=[circle,fill=white,draw=black, inner sep = 2.5]

\begin{figure}[htb]
\centering
\captionsetup[subfloat]{labelformat=empty, justification=Centering}

\newcommand{\wheelsascalar}{1.2}
\subfloat[bipyramid\label{subfig-1}]{%
 \begin{tikzpicture}
 
 \node[style=normalvertex,label=right:{}] (v-1) at (-40:\wheelsascalar cm) {};
  \node[style=normalvertex,label=right:{}] (v0) at (0:\wheelsascalar cm) {};
  \node[style=normalvertex,label=right:{}] (v1) at (40:\wheelsascalar cm) {};
  \node[style=normalvertex,label=right:{}] (v2) at (80:\wheelsascalar cm) {};
  \node[style=normalvertex,label=right:{}] (v3) at (120:\wheelsascalar cm) {};
  \node[style=normalvertex,label=right:{}] (v4) at (160:\wheelsascalar cm) {};
  \node[style=normalvertex,label=right:{}] (v5) at (200:\wheelsascalar cm) {};
  \node[style=normalvertex,label=right:{}] (v6) at (240:\wheelsascalar cm) {};
 
  \node[style=normalvertex,label=right:{}] (cent) at (0,0) {};
  \node[style=normalvertex,label=right:{}] (out) at (-0.1,2.25) {};
	 	 
  \foreach \x  in {-1,0,...,6}
  {
   \path (cent) edge (v\x);
  }
  \foreach \x/\y in {-1/0,0/1,1/2,2/3,3/4,4/5,5/6}
  {
   \path (v\x) edge (v\y);
  }

\draw[fill]  (267:1.0 cm) ellipse (0.06 and 0.06);
\draw[fill]  (280:1.0 cm) ellipse (0.06 and 0.06);
\draw[fill]  (293:1.0 cm) ellipse (0.06 and 0.06);

\draw (out) .. controls (1.5,1.5) and (1.75,1) .. (v0);
	 \path  (out) edge (v1) (out) edge (v2) (out) edge (v3);
	 \draw (out) .. controls (-1,1.5) and (-1.5,1) .. (v4);	
	 \draw (out) .. controls (-2,1.5) and (-2,0) .. (v5);	 
	  \draw (out) .. controls (-3,2) and (-3,-1) .. (v6);	 
	   \draw (out) .. controls (2.5,2) and (2.5,-0.5) .. (v-1);	 
\end{tikzpicture}
}
\newcommand{\tetrasascalar}{1.5}
\newcommand{\tetrasascalartwo}{0.4}
\subfloat[tetrahedron\label{subfig-2}]{%
 \begin{tikzpicture}
 
  \node[style=normalvertex,label=right:{}] (r) at (-30:\tetrasascalar cm) {};
  \node[style=normalvertex,label=right:{}] (u) at (90:\tetrasascalar cm) {};
  \node[style=normalvertex,label=right:{}] (l) at (-150:\tetrasascalar cm) {};
 
  \node[style=normalvertex,label=right:{}] (cent) at (0,0) {};

  \path
   (r) edge (u) (u) edge (l) (l) edge (r)
   (r) edge (cent) (l) edge (cent) (u) edge (cent);
\end{tikzpicture}
}
\subfloat[triakis tetrahedron\label{subfig-3}]{%
 \begin{tikzpicture}

  \node[style=normalvertex,label=right:{}] (r) at (-30:\tetrasascalar cm) {};
  \node[style=normalvertex,label=right:{}] (u) at (90:\tetrasascalar cm) {};
  \node[style=normalvertex,label=right:{}] (l) at (-150:\tetrasascalar cm) {};
  
  \node[style=normalvertex,label=right:{}] (d2) at (-90:\tetrasascalartwo cm) {};
  \node[style=normalvertex,label=right:{}] (ru2) at (30:\tetrasascalartwo  cm) {};
  \node[style=normalvertex,label=right:{}] (lu2) at (150:\tetrasascalartwo cm) {};
 
  \node[style=normalvertex,label=right:{}] (cent) at (0,0) {};
  \node[style=normalvertex,label=right:{}] (out) at (0,2) {};

  \path
   (r) edge (u) (u) edge (l) (l) edge (r)
   (r) edge (cent) (l) edge (cent) (u) edge (cent)
      
   (lu2) edge (u) (lu2) edge (l) (lu2) edge (cent)
   (ru2) edge (u) (ru2) edge (r) (ru2) edge (cent)
   (d2) edge (l) (d2) edge (r) (d2) edge (cent)
  ;
	 \draw (out) .. controls (-1.5,1.75) and (-1.5,-0.3) .. (l);
	 \draw (out) .. controls (1.5,1.75) and (1.5,-0.3) .. (r);	 
	 \path (out) edge (u);
\end{tikzpicture}
}

\subfloat[cube (hexahedron)\label{subfig-4}]{%
	\begin{tikzpicture}
	  \node[style=normalvertex,label=right:{}] (ld) at (-\innerscalar,-\innerscalar) {};
	  \node[style=normalvertex,label=right:{}] (lu) at (-\innerscalar,\innerscalar) {};
	  \node[style=normalvertex,label=right:{}] (rd) at (\innerscalar,-\innerscalar) {};
	  \node[style=normalvertex,label=right:{}] (ru) at (\innerscalar,\innerscalar) {};
	
	 \node[style=normalvertex,label=right:{}] (ld2) at (-\outerscalar,-\outerscalar) {};
	  \node[style=normalvertex,label=right:{}] (lu2) at (-\outerscalar,\outerscalar) {};
	  \node[style=normalvertex,label=right:{}] (rd2) at (\outerscalar,-\outerscalar) {};
	  \node[style=normalvertex,label=right:{}] (ru2) at (\outerscalar,\outerscalar) {};
	 
	  \path
	   (lu) edge (ld) (ld) edge (rd) (rd) edge (ru)  (ru) edge (lu)
	   
	   (lu2) edge (ld2) (ld2) edge (rd2) (rd2) edge (ru2) (ru2) edge (lu2)
	  
	   (lu) edge (lu2) (ld) edge (ld2) (ru) edge (ru2) (rd) edge (rd2)  
	   ;

	\end{tikzpicture}
}
\subfloat[tetrakis hexahedron\label{subfig-5}]{%
	\begin{tikzpicture}
	  \node[style=normalvertex,label=right:{}] (ld) at (-\innerscalar,-\innerscalar) {};
	  \node[style=normalvertex,label=right:{}] (lu) at (-\innerscalar,\innerscalar) {};
	  \node[style=normalvertex,label=right:{}] (rd) at (\innerscalar,-\innerscalar) {};
	  \node[style=normalvertex,label=right:{}] (ru) at (\innerscalar,\innerscalar) {};
	
	 \node[style=normalvertex,label=right:{}] (ld2) at (-\outerscalar,-\outerscalar) {};
	  \node[style=normalvertex,label=right:{}] (lu2) at (-\outerscalar,\outerscalar) {};
	  \node[style=normalvertex,label=right:{}] (rd2) at (\outerscalar,-\outerscalar) {};
	  \node[style=normalvertex,label=right:{}] (ru2) at (\outerscalar,\outerscalar) {};
	 
	 \node[style=normalvertex,label=right:{}] (ju) at (0,\middlescalar) {};
	 \node[style=normalvertex,label=right:{}] (jl) at (-\middlescalar,0) {};
	 \node[style=normalvertex,label=right:{}] (jd) at (0,-\middlescalar) {};
	 \node[style=normalvertex,label=right:{}] (jr) at (\middlescalar,0) {};
	 
	 \node[style=normalvertex,label=right:{}] (out) at (0,\outvertexh) {};
	 
	 \node[style=normalvertex,label=right:{}] (cent) at (0,0) {};
	 
	  \path
	   (lu) edge (ld) (ld) edge (rd) (rd) edge (ru)  (ru) edge (lu)
	   
	   (lu2) edge (ld2) (ld2) edge (rd2) (rd2) edge (ru2) (ru2) edge (lu2)
	  
	   (lu) edge (lu2) (ld) edge (ld2) (ru) edge (ru2) (rd) edge (rd2) 
	   
	   (lu) edge (ju) (ru) edge (ju) (lu2) edge (ju) (ru2) edge (ju)
	   (ld) edge (jd) (rd) edge (jd) (ld2) edge (jd) (rd2) edge (jd)  
	   (ld) edge (jl) (lu) edge (jl) (ld2) edge (jl) (lu2) edge (jl)  
	   (rd) edge (jr) (ru) edge (jr) (rd2) edge (jr) (ru2) edge (jr)  
	   
	   (ru) edge (cent) (rd) edge (cent) (lu) edge (cent) (ld) edge (cent)  
	   ;
	   
 \draw (out) .. controls (-1,1.5) and (-1.1,1.2) .. (lu2);
 \draw (out) .. controls (1,1.5) and (1.1,1.2) .. (ru2);
 \draw (out) .. controls (-2,2.5) and (-2,-1.1) .. (ld2);
 \draw (out) .. controls (2,2.5) and (2,-1.1) .. (rd2);	   
	   
	\end{tikzpicture}
}
\subfloat[rhombic dodecahedron\label{subfig-6}]{%
	\begin{tikzpicture}
	  \node[style=normalvertex,label=right:{}] (ld) at (-\innerscalar,-\innerscalar) {};
	  \node[style=normalvertex,label=right:{}] (lu) at (-\innerscalar,\innerscalar) {};
	  \node[style=normalvertex,label=right:{}] (rd) at (\innerscalar,-\innerscalar) {};
	  \node[style=normalvertex,label=right:{}] (ru) at (\innerscalar,\innerscalar) {};
	
	 \node[style=normalvertex,label=right:{}] (ld2) at (-\outerscalar,-\outerscalar) {};
	  \node[style=normalvertex,label=right:{}] (lu2) at (-\outerscalar,\outerscalar) {};
	  \node[style=normalvertex,label=right:{}] (rd2) at (\outerscalar,-\outerscalar) {};
	  \node[style=normalvertex,label=right:{}] (ru2) at (\outerscalar,\outerscalar) {};
	 
	 \node[style=normalvertex,label=right:{}] (ju) at (0,\middlescalar) {};
	 \node[style=normalvertex,label=right:{}] (jl) at (-\middlescalar,0) {};
	 \node[style=normalvertex,label=right:{}] (jd) at (0,-\middlescalar) {};
	 \node[style=normalvertex,label=right:{}] (jr) at (\middlescalar,0) {};
	 
	 \node[style=normalvertex,label=right:{}] (out) at (0,\outvertexh) {};
	 
	 \node[style=normalvertex,label=right:{}] (cent) at (0,0) {};
	 
	  \path
	   
	   (lu) edge (ju) (ru) edge (ju) (lu2) edge (ju) (ru2) edge (ju)
	   (ld) edge (jd) (rd) edge (jd) (ld2) edge (jd) (rd2) edge (jd)  
	   (ld) edge (jl) (lu) edge (jl) (ld2) edge (jl) (lu2) edge (jl)  
	   (rd) edge (jr) (ru) edge (jr) (rd2) edge (jr) (ru2) edge (jr)  
	   
	   (ru) edge (cent) (rd) edge (cent) (lu) edge (cent) (ld) edge (cent)  
	   ;
	   
 \draw (out) .. controls (-1,1.5) and (-1.1,1.2) .. (lu2);
 \draw (out) .. controls (1,1.5) and (1.1,1.2) .. (ru2);
 \draw (out) .. controls (-2,2.5) and (-2,-1.1) .. (ld2);
 \draw (out) .. controls (2,2.5) and (2,-1.1) .. (rd2);	 		 
	\end{tikzpicture}
}
\newcommand{\icosascalar}{2.1}
\newcommand{\icosascalartwo}{0.7}
\newcommand{\icosascalarthree}{0.25}
\subfloat[icosahedron\label{subfig-7}]{%
 \begin{tikzpicture}

  \node[style=normalvertex,label=right:{}] (r) at (-30:\icosascalar cm) {};
  \node[style=normalvertex,label=right:{}] (u) at (90:\icosascalar cm) {};
  \node[style=normalvertex,label=right:{}] (l) at (-150:\icosascalar cm) {};
  
  \node[style=normalvertex,label=right:{}] (d2) at (-90:\icosascalartwo cm) {};
  \node[style=normalvertex,label=right:{}] (u2) at (90:\icosascalartwo cm) {};
  \node[style=normalvertex,label=right:{}] (ru2) at (30:\icosascalartwo cm) {};
  \node[style=normalvertex,label=right:{}] (rd2) at (-30:\icosascalartwo cm) {};
  \node[style=normalvertex,label=right:{}] (lu2) at (150:\icosascalartwo cm) {};
  \node[style=normalvertex,label=right:{}] (ld2) at (-150:\icosascalartwo cm) {};
  
  \node[style=normalvertex,label=right:{}] (r3) at (30:\icosascalarthree cm) {};
  \node[style=normalvertex,label=right:{}] (d3) at (-90:\icosascalarthree cm) {};
  \node[style=normalvertex,label=right:{}] (l3) at (150:\icosascalarthree cm) {};
  
  \path
   (r) edge (u) (u) edge (l) (l) edge (r)
   (r3) edge (d3) (d3) edge (l3) (l3) edge (r3)

(d2) edge (ld2) (ld2) edge (lu2) (lu2) edge (u2) (u2) edge (ru2) (ru2) edge (rd2) (rd2) edge (d2)
   
   (r) edge (rd2) (r) edge (d2) (r) edge (ru2)
   (l) edge (ld2) (l) edge (d2) (l) edge (lu2)
   (u) edge (lu2) (u) edge (ru2) (u) edge (u2)
   
   (r3) edge (ru2) (r3) edge (u2) (r3) edge (rd2)
   (l3) edge (ld2) (l3) edge (u2) (l3) edge (lu2)
   (d3) edge (ld2) (d3) edge (rd2) (d3) edge (d2)   
  ;
\end{tikzpicture}
}

\newcommand{\octasascalar}{2.7}
\newcommand{\octasascalartwo}{0.73}
\newcommand{\octasascalartwopoint}{0.95}
\newcommand{\octasascalarthree}{0.5}
\subfloat[octahedron\label{subfig-8}]{%
 \begin{tikzpicture}
 
  \node[style=normalvertex,label=right:{}] (r) at (-30:\octasascalar cm) {};
  \node[style=normalvertex,label=right:{}] (u) at (90:\octasascalar cm) {};
  \node[style=normalvertex,label=right:{}] (l) at (-150:\octasascalar cm) {};
  
  \node[style=normalvertex,label=right:{}] (ru3) at (30:\octasascalarthree cm) {};
  \node[style=normalvertex,label=right:{}] (lu3) at (150:\octasascalarthree cm) {};
  \node[style=normalvertex,label=right:{}] (d3) at (-90:\octasascalarthree cm) {};
  \path
   (r) edge (u) (u) edge (l) (l) edge (r)
   (ru3) edge (lu3) (lu3) edge (d3) (d3) edge (ru3)
   
   (r) edge (ru3) (r) edge (d3) (l) edge (lu3) (l) edge (d3) (u) edge (ru3) (u) edge (lu3) 
  ;
\end{tikzpicture}
}
\subfloat[triakis octahedron\label{subfig-9}]{%
 \begin{tikzpicture}
 
  \node[style=normalvertex,label=right:{}] (r) at (-30:\octasascalar cm) {};
  \node[style=normalvertex,label=right:{}] (u) at (90:\octasascalar cm) {};
  \node[style=normalvertex,label=right:{}] (l) at (-150:\octasascalar cm) {};
  
  \node[style=normalvertex,label=right:{}] (d2) at (-90:\octasascalartwopoint cm) {};
  \node[style=normalvertex,label=right:{}] (u2) at (90:\octasascalartwo cm) {};
  \node[style=normalvertex,label=right:{}] (ru2) at (30:\octasascalartwopoint cm) {};
  \node[style=normalvertex,label=right:{}] (rd2) at (-30:\octasascalartwo cm) {};
  \node[style=normalvertex,label=right:{}] (lu2) at (150:\octasascalartwopoint cm) {};
  \node[style=normalvertex,label=right:{}] (ld2) at (-150:\octasascalartwo cm) {};
  
  \node[style=normalvertex,label=right:{}] (ru3) at (30:\octasascalarthree cm) {};
  \node[style=normalvertex,label=right:{}] (lu3) at (150:\octasascalarthree cm) {};
  \node[style=normalvertex,label=right:{}] (d3) at (-90:\octasascalarthree cm) {};

	 \node[style=normalvertex,label=right:{}] (cent) at (0,0) {};
	 
	 	 \node[style=normalvertex,label=right:{}] (out) at (0,3.25) {};

  \path
   (r) edge (u) (u) edge (l) (l) edge (r)
   (ru3) edge (lu3) (lu3) edge (d3) (d3) edge (ru3)
   
   (r) edge (ru3) (r) edge (d3) (l) edge (lu3) (l) edge (d3) (u) edge (ru3) (u) edge (lu3) 
   
   (lu2) edge (u) (lu2) edge (l) (lu2) edge (lu3)
   (ru2) edge (u) (ru2) edge (r) (ru2) edge (ru3)
   (d2) edge (l) (d2) edge (r) (d2) edge (d3)

   (ld2) edge (d3) (ld2) edge (l) (ld2) edge (lu3)
   (rd2) edge (d3) (rd2) edge (r) (rd2) edge (ru3)
   (u2) edge (u) (u2) edge (ru3) (u2) edge (lu3)
   
   (cent) edge (lu3) (cent) edge (ru3) (cent) edge (d3)
   
  ;
	 \draw (out) .. controls (-3,3) and (-3,-1) .. (l);
	 \draw (out) .. controls (3,3) and (3,-1) .. (r);	 
	 \path (out) edge (u);
\end{tikzpicture}
}
\caption{The~3-connected planar graphs with fixing number~3.}
\label{fig:3:con:plan:gra:fix:3}
\end{figure}
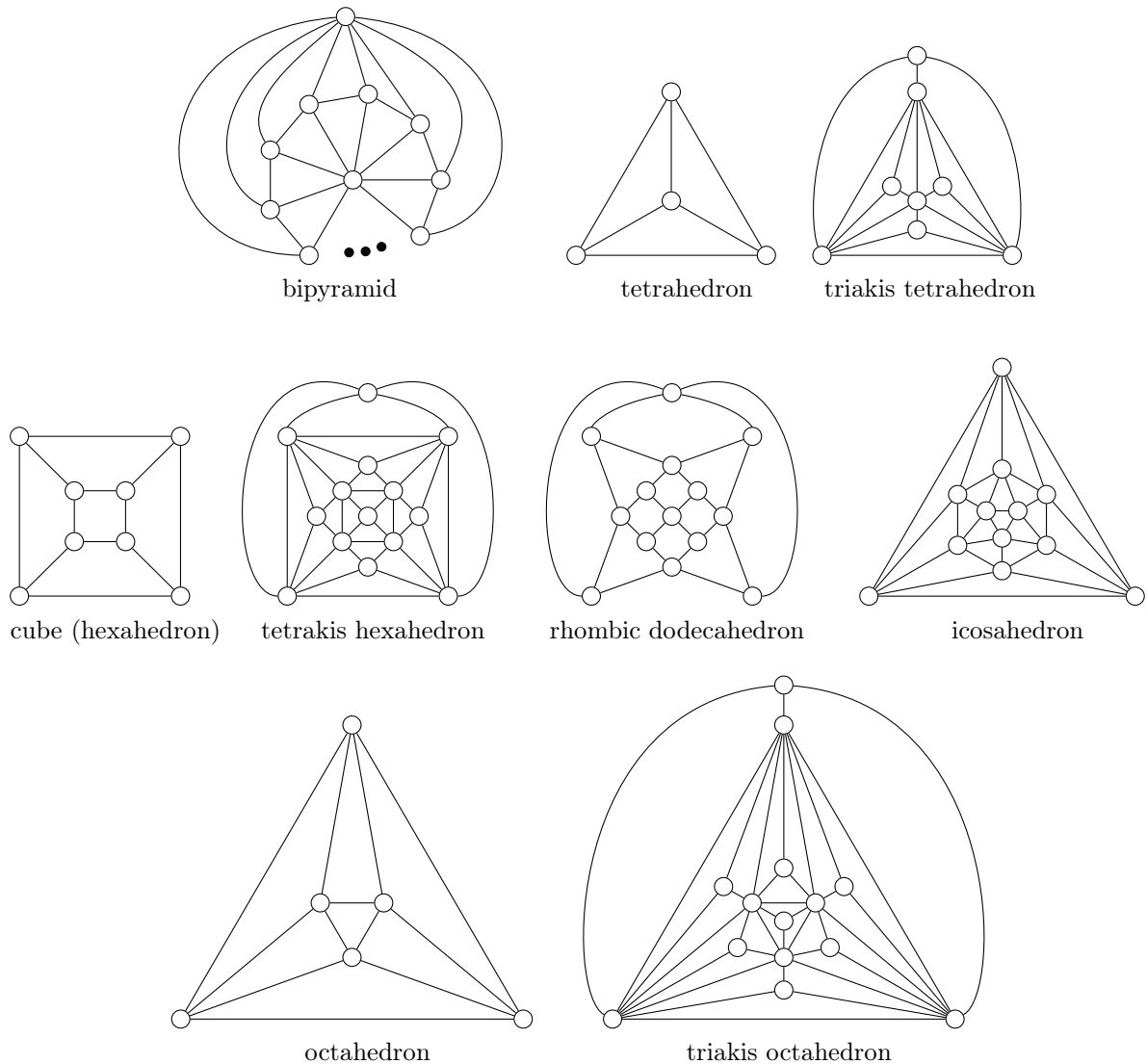

\begin{table}[ht]
	\caption{Properties of the planar~3-connected graphs with fixing number~3.}\label{t1}
	\begin{center}
		\begin{tabular}{|c|c|c|c|c|c|}
			\hline 
			\rule[-1ex]{0pt}{2.5ex} Name & |V| & |E| & |F| & V-type & F-type  \\ 
			\hline 
			\rule[-1ex]{0pt}{2.5ex}~$n$-bipyramid ($n\ge 3$)&~$n+2$ &~$3n$ &~$2n$ &~$2\{n\}+n\{4\}$ &~$2n\{3\}$  \\ 
			\hline 
			\rule[-1ex]{0pt}{2.5ex} tetrahedron &~$4$ &~$6$ &~$4$ &~$4\{3\}$ &~$4\{3\}$    \\ 
			\hline 
			\rule[-1ex]{0pt}{2.5ex} cube &~$8$ &~$12$ &~$6$ &~$8\{3\}$ &~$6\{4\}$    \\ 
			\hline 
			\rule[-1ex]{0pt}{2.5ex} triakis tetrahedron &~$8$ &~$18$ &~$12$ &~$4\{3\}+4\{6\}$ &~$12\{3\}$   \\ 
			\hline 
			\rule[-1ex]{0pt}{2.5ex} icosahedron &~$12$ &~$30$ &~$20$ &~$12\{5\}$ &~$20\{3\}$   \\ 
			\hline 
			\rule[-1ex]{0pt}{2.5ex} rhombic dodecahedron &~$14$ &~$24$ &~$12$ &~$8\{3\}+6\{4\}$ &~$12\{4\}$ \\ 
			\hline 
			\rule[-1ex]{0pt}{2.5ex} triakis octahedron &~$14$ &~$36$ &~$24$ &~$8\{3\}+6\{8\}$ &~$24\{3\}$  \\ 
			\hline 
			\rule[-1ex]{0pt}{2.5ex} tetrakis hexahedron &~$14$ &~$36$ &~$24$ &~$6\{4\}+8\{6\}$ &~$24\{3\}$ \\ 
			\hline 
		\end{tabular}
	\end{center}
\end{table}

On the algorithmic side we obtain a very easy algorithm to check isomorphism of~3-connected planar graphs. In fact the arguments show that with the right cell selection strategy, individualization refinement algorithms (such as nauty and traces~\cite{mckay} or bliss \cite{bliss}), which constitute to-date the fastest isomorphism algorithms in practice, have polynomial running time on~3-connected planar graphs.

Concerning lower bounds on the WL-dimension of planar graphs, it is not difficult to see that there are planar graphs with WL-dimension~2 (for example the~6-cycle). However, the question whether the maximum dimension of the class of planar graphs is~2 or~3 remains open.

\paragraph*{Related work.}

There is an extensive body of work on isomorphism testing of planar graphs. Most notably Hopcroft and Tarjan first exploited the decomposition of a graph into its~3-connected components to obtain an algorithm with quasi-linear running time~\cite{HOPCROFT1973323,DBLP:journals/ipl/HopcroftT71,DBLP:conf/coco/HopcroftT72}, which led to a linear-time algorithm by Hopcroft and Wong~\cite{DBLP:conf/stoc/HopcroftW74}. 
More recent results show that isomorphism of planar graphs can be decided in logarithmic space~\cite{DBLP:conf/coco/DattaLNTW09}.

There are also various results on the descriptive complexity of planar graph isomorphism. In this direction, Grohe shows that fixed-point logic with counting FPC captures polynomial time on planar graphs~\cite{DBLP:conf/lics/Grohe98} and more generally on graphs of bounded genus~\cite{DBLP:conf/stoc/Grohe00}. This was also known for graphs of bounded tree width~\cite{DBLP:conf/icdt/GroheM99}.

Subsequent work shows that for~3-connected planar graphs and for graphs of bounded tree width it is possible to restrict the quantifier depth (or equivalently, the number of iterations that the WL-algorithm performs until it terminates) to a polylogarithmic number, which translates to parallel isomorphism tests~\cite{DBLP:conf/icalp/GroheV06,DBLP:conf/stacs/Verbitsky07}. For general graphs, recent results give new upper and lower bounds on the quantifier depth (translating into bounds on the maximum number of iterations of the WL-algorithm) 
\cite{DBLP:conf/lics/KieferS16a,DBLP:conf/lics/BerkholzN16}.
Extending the results on planar graphs in the direction of dynamic complexity, Mehta shows that isomorphism of~3-connected planar graphs is in DynFO+~\cite{DBLP:conf/csr/Mehta14}, where in fact no counting quantifiers are required.

While it is possible to describe precisely the graphs of WL-dimension~1~\cite{DBLP:conf/mfcs/KieferSS15,DBLP:conf/fct/ArvindKRV15} (i.e., graphs definable with a~2-variable sentence in first order logic with counting), it appears difficult to make such statements for higher dimensions. However, for various graph classes for which the isomorphism problem is known to be polynomial time solvable, one can give upper bounds on the dimension. 
E.g., for cographs, interval graphs, and, more generally, for rooted directed-path graphs, it suffices to apply the~2-dimensional WL-algorithm in order to decide isomorphism~\cite{DBLP:journals/combinatorics/EvdokimovP00}. In general, isomorphism of graph classes with an excluded minor can be solved in polynomial time \cite{Ponomarenko} and in fact a sufficiently high-dimensional WL-algorithm will decide isomorphism on such a class~\cite{DBLP:journals/jacm/Grohe12,grohe_2017}. More strongly, FPC captures polynomial time on graph classes with an excluded minor. 
In the proof of this result, structural graph theory and in particular decompositions play a central role. While our paper  uses very basic parts of these techniques and concepts, they are only implicit and we refer the reader to~\cite{DBLP:journals/corr/Grohe16b} for a more systematic treatment.

\section{Preliminaries}

All graphs in this paper are finite simple graphs, that is, undirected graphs without loops. The vertex and the edge set of a graph~$G$ are denoted by~$V(G)$ and~$E(G)$, respectively. The neighborhood~$N(X)$ of a subset of the vertices~$X \subseteq V(G)$ is the set~$\{ u\in V(G)\setminus X \mid \exists v \in X \text{ s.t.\ } \{u,v\} \in E(G)\}$. 
For~$X \subseteq V(G)$, we denote by~$G[X]$ the \emph{subgraph of~$G$ induced by~$X$}, i.e., the graph with vertex set~$X$ and edge set~$E(G) \cap \{ \{ u,v \} \mid u,v \in X\}$.
The graph~$G-X \coloneqq G[V(G)\setminus X]$ is obtained from~$G$ by removing~$X$. We write~$G \cong H$ to indicate that~$G$ is \emph{isomorphic} to~$H$.
A \emph{minor} of~$G$ is a graph obtained by repeated vertex deletions, edge deletions and edge contractions.

For a positive integer~$k$, a graph~$G$ is \emph{$k$-connected} if~$G$ has more than~$k$ vertices and for all~$X \subseteq V(G)$ with~$|X| < k$, the graph~$G - X$ is connected.
A \emph{separator}~$S\subseteq V(G)$ is a subset of the vertices such that~$G-S$ is not connected. A vertex~$v$ is a \emph{cut vertex} if~$\{v\}$ is a separator, and a \emph{2-separator} is a separator of size~2. 
A \emph{2-connected component} of~$G$ is a subset~$S'$ of~$V(G)$ such that the graph~$G[S']$ is~$2$-connected and such that~$S'$ is maximal with respect to inclusion. We refer the reader to~\cite{MR1844449} for more basic information on graphs, in particular on \emph{planar graphs}, which are graphs that can be drawn in the plane without crossings.

A \emph{vertex-colored graph}~$(G, \lambda)$ is a graph~$G$ with a function~$\lambda \colon V \rightarrow \mathcal{C}$, where~$\mathcal{C}$ is an arbitrary set. We call~$\lambda$ a \emph{vertex coloring} of~$G$. Similarly, an \emph{arc-colored graph} is a graph~$G$ with a function~$\lambda \colon \{(u,u)\mid u\in V(G)\} \cup \{ (u,v) \mid \{u,v\} \in E(G)\}\rightarrow \mathcal{C}$. In this case, we call~$\lambda$ an \emph{arc coloring}.
We interpret~$\lambda(u,u)$ as the vertex color of~$u$ and for~$\{u,v\}\in E(G)$ we interpret~$\lambda(u,v)$ as the color of the arc from~$u$ to~$v$. In particular it may be the case that~$\lambda(u,v) \neq \lambda(v,u)$. However, while we allow such colorings, all graphs in this paper are undirected. Furthermore, we treat every uncolored graph as a monochromatic colored graph.

\paragraph*{The Weisfeiler-Leman algorithm (see \cite{DBLP:journals/combinatorica/CaiFI92}).} For~$k \in \mathbb{N}$, a graph~$G$ and a coloring~$\lambda$ of~$V^k(G)$, let~$(v_1, \dots, v_k)$ be a vertex~$k$-tuple of~$G$. We define~$\prescript{}{0}\chi_G^k (v_1, \dots, v_k)$ to be a tuple consisting of an encoding of~$\lambda(v_1, \dots, v_k)$ and an encoding of the isomorphism class of the colored graph obtained from~$G[\{v_1, \dots, v_k\}]$ by coloring for~$i \in \{1, \dots, k\}$ vertex~$v_i$ with color~$i$.

That is, for a second graph~$G'$, possibly equal to~$G$, with coloring~$\lambda'$ and for a vertex~$k$-tuple~$(v_1',\dots,v_k')$ of~$G'$ we have
\[\prescript{}{0} \chi^k_G (v_1, \dots, v_k) = \prescript{}{0} \chi^k_{G'} (v_1', \dots, v_k')\]
if and only if~$\lambda (v_1, \dots, v_k) = \lambda'  (v_1', \dots, v_k')$ and there is an isomorphism from~$G[v_1,\dots, v_k]$ to~$G'[v'_1,\dots, v'_k]$ mapping~$v_j$ to~$v'_j$ for all~$j\in \{1,\dots,k\}$. 

We recursively define the color~$\prescript{}{i+1} \chi^k_G (v_1, \dots, v_k)$ by setting
\[\prescript{}{i+1} \chi^k_G (v_1, \dots, v_k) \coloneqq (\prescript{}{i} \chi^k_G (v_1, \dots, v_k); \mathcal{M}),
\]
where~$\mathcal{M}$ is the multiset defined as
\[{\mathcal{M}}:= \big\{\!\!\big\{\big(\prescript{}{i} \chi^k_G(w,v_2, \dots, v_k), \dots, \prescript{}{i} \chi^k_G(v_1, \dots, v_{k-1}, w)
\big) \ \big\vert \ w\in V     \big\}\!\!\big\}\]
if~$k \geq 2$ and as~${\mathcal{M}}:=  \{\!\! \{ (\prescript{}{i} \chi^1_G(w) \mid w\in N(v_1) \}\!\!\}$ if~$k = 1$. That is, if~$k = 1$, the iteration is only over neighbors of~$v_1$.

There is a slight technical issue about the initial coloring. Suppose for a fixed~$k$ we are given a graph~$G$ with a coloring~$\lambda'$ of~$V^\ell(G)$ for~$\ell < k$. To turn it into a correct input for the~$k$-dimensional WL-algorithm, we replace~$\lambda'$ by an appropriate coloring~$\lambda$. For a vertex tuple~$(u_1, \dots, u_k)$, we define~$\lambda (u_1, \dots, u_k) \coloneqq  \lambda'(u_1, \dots ,u_{\ell})$. Note that~$\lambda$ preserves all information from~$\lambda'$. If~$\lambda'$ is an arc coloring we define~$\lambda (u_1, \dots, u_k)$ to be~$(\lambda'(u_1, u_2),0)$ if~$(u_1,u_2)$ is in the domain of~$\lambda'$ and to be~$(1,1)$ otherwise.

By definition, the coloring~$\prescript{}{i+1} \chi^k_G$ induces a refinement of the partition of the~$k$-tuples of the vertices of the graph~$G$ with coloring~$\prescript{}{i}\chi^k_G$. Thus, there is some minimal~$i$ such that the partition induced by the coloring~$\prescript{}{i+1} \chi^k_G$ is not strictly finer than the one induced by the coloring~$\prescript{}{i} \chi^k_G$ on~$G$. For this minimal~$i$, we call~$\prescript{}{i} \chi^k_G$ the \emph{stable} coloring of~$G$ and denote it by~$\chi^k_G$. 

For~$k \in \mathbb{N}$, the \emph{$k$-dimensional WL-algorithm} takes as input a vertex coloring or an arc coloring~$\lambda$ of a graph~$G$ and returns the coloring~$\chi^k_G$. For two graphs~$G$ and~$G'$, we say that the~$k$-dimensional WL-algorithm \emph{distinguishes}~$G$ and~$G'$ if its application to each of them results in colorings with differing color class sizes. More precisely, the graphs~$G$ and~$G'$ are distinguished if there is a color~$C$ in the range of~$\chi^k_G$ such that the sets
$\{\bar{v} \mid \bar{v} \in V^k(G), \chi^k_G(\bar{v}) = C\}$ and~$\{\bar{w} \mid \bar{w} \in V^k(G'), \chi^k_{G'}(\bar{w}) = C\}$ have different cardinalities.

If two graphs are distinguished by the~$k$-dimensional WL-algorithm for some~$k$, then they are not isomorphic. However, if~$k$ is fixed, the converse is not always true. There is a close connection between the WL-algorithm and first-order logic with counting (as well as fixed-point logic with counting). We refer the reader to existing literature (for example~\cite{DBLP:journals/combinatorica/CaiFI92,grohe_2017}) for more information.

For improved readability, we will use the letter~$\lambda$ to denote arbitrary colorings that do not necessarily result from applications of the WL-algorithm.

\section{Decompositions}\label{sec:decomps}

For a graph~$G$, define the set~$P(G)$ to consist of
the pairs~$(S,K)$, where~$S$ is a separator of~$G$ of minimum cardinality and~$K \subseteq V(G)\setminus S$ is the vertex set of a connected component of~$G-S$.

We observe that if~$G$ is a connected graph that is not~$2$-connected, then~$P(G)$ is the set of pairs~$(\{s\},K)$ where~$s$ is a cut vertex and~$G[K]$ a connected component of~$G- \{s\}$. In this case we also write~$(s,K)$ instead of~$(\{s\},K)$. If~$G$ is~2-connected but not~3-connected, all separators in~$P(G)$ have size~2.

There is a natural partial order on~$P(G)$ with respect to inclusion in the second component, i.e., we can define:
\[
(S,K) \leq (S',K') \iff K \subseteq K'.  
\]

We define~$P_0(G)$ to be the set of minimal elements of~$P(G)$ with respect to this partial order.

\begin{remark}\label{remark:isomorphism_invariant}
It immediately follows from the definitions that the sets~$P(G)$ and~$P_0(G)$ (and the corresponding partial orders) are isomorphism-invariant (i.e., preserved under isomorphisms).
\end{remark}

 Note that~$P_0(G)$ is non-empty whenever~$G$ is not a complete graph. Also note that if~$G$ is not~2-connected, then for two distinct minimal elements~$(S,K)$ and~$(S',K')$ in~$P_0(G)$ we have~$K \cap K' = \varnothing$. Furthermore, in the case that~$G$ is connected but not~$2$-connected, the set~$P_0(G)$ contains exactly the pairs~$(s,K)$ for which~$s$ is a cut vertex and~$G[K]$ a connected component of~$G-\{s\}$ that does not contain a cut vertex of~$G$. 
These two observations can be generalized to graphs with a higher connectivity, but for this we need an additional requirement on the minimum degree as follows.

\begin{lemma}\label{lem:minimum:degree:implies:disjoin:min:seps}
Let~$G$ be a graph that is not~$(k+1)$-connected and has minimum degree at least~$\frac{3k-1}{2}$. 
\begin{enumerate}
\item \label{item:one:min:comp:trivial:int:contained}
If~$(S,K)\in P_0(G)$ and~$(S',K')\in P(G)$ are distinct, then~$K\subseteq K'$ or~$(K\cup S)\cap (K'\cup S') = S\cap S'$.
\item If~$(S,K),(S',K') \in P_0(G)$ are distinct, then~$(K\cup S)\cap (K'\cup S') = S\cap S'$.\label{item:min:comp:trivial:int}
\item A pair~$(S,K)\in P(G)$ is contained in~$P_0(G)$ if and only if there is no separator~$S'$ of~$G$ of minimum cardinality with~$S'\cap K \neq \varnothing$.\label{item:char:min:comp}
\end{enumerate}
\end{lemma}

\begin{proof} (Part~\ref{item:one:min:comp:trivial:int:contained})
Assume that~$(S,K)\in P_0(G)$ and~$(S',K')\in P(G)$ are distinct. Note that~$S = N(K)$ and~$S' = N(K')$ since~$S$ and~$S'$ are minimal separators.

Suppose~$K \not\subseteq K'$ and that there exists~$v \in K \cap K'$. Then~$v$ has a neighbor~$u \in K$ which does not belong to~$K'$. Since~$v \in K'$, this implies that~$u$ belongs to~$S'$. Therefore, the graph~$G-u$ is at most~$(k-1)$-connected. On the other hand,~$u$ lies in~$K$. By Corollary~$1$ in \cite{Karpov}, the graph~$G-u$ is~$k$-connected, yielding a contradiction.

(Part~\ref{item:min:comp:trivial:int}) This follows by applying Part~\ref{item:one:min:comp:trivial:int:contained} twice.

(Part~\ref{item:char:min:comp}) If~$(S,K)\in P(G)$ is not minimal then there is~$(S',K')\in P(G)$ with~$K'\subsetneq K$. Then~$S'\subseteq K\cup S$ but~$S\neq S'$, which shows that~$S'\cap K\neq \varnothing$. Conversely, suppose~$(S,K)\in P_0(G)$ and that there is a minimum size separator~$S'$ of~$G$ with~$S'\cap K\neq \varnothing$. Let~$K' \subseteq V(G-S')$ be a vertex set such that~$G[K']$ is a connected component of~$G-S'$.
Then~$(S,K)$ and~$(S',K')$ violate Part~\ref{item:one:min:comp:trivial:int:contained} of the lemma.
\end{proof}

We remark that for a connected graph~$G$ which is not~$3$-connected, the elements of~$P_0(G)$ correspond to the leaves in a suitable decomposition tree (i.e., the decomposition into~2- or~3-connected components) in the sense of Tutte. 
However, we will not require this fact.

In the following we present a method to remove the vertices appearing in the second component of pairs in~$P_0$ from graphs in such a way that the property whether two graphs are isomorphic is preserved. This will allow us to devise an inductive isomorphism test. In the next sections we will then show that a sufficiently high-dimensional WL-algorithm in some sense implicitly performs this induction.

For a graph~$G$ and a set~$S \subseteq V(G)$, we define the graph~$G^S_{\top}$ as consisting of the vertices of~$S$ and those appearing together with~$S$ in~$P_0(G)$. More precisely,
$G^S_{\top}$ is the graph on the vertex set~\[V'= S \cup \bigcup\limits_{(S,K)\in P_0(G)} K\]
 with edge set~$E' = E(G[V']) \cup \{\{s,s'\}\mid s,s'\in S \text{ and }s\neq s'\}$, see Figure~\ref{fig:decomposed:graphs} left and bottom right.\footnote{For the reader familiar with tree decompositions we remark that this graph corresponds to the torso of the bag~$S\cup K$ in a suitable tree decomposition. However, we will not require this point of view in the paper.} Note that if~$S$ is not a separator or does not appear as a separator in~$P_0(G)$, then~$V'=S$.

\tikzset{invclip/.style={clip,insert path={{[reset cm]
      (-16383.99999pt,-16383.99999pt) rectangle (16383.99999pt,16383.99999pt)
    }}}}
\begin{figure}
\centering
\begin{tikzpicture}[scale = 0.8]

\draw[rotate=30]  (1,0) ellipse (2 and 1.3);

\draw[rotate=-30]  (3,2.5) ellipse (2 and 1.1);
\node[circle, draw, fill, minimum size = 5pt, inner sep = 0, outer sep = 0] (v5) at (2.34,1.4) {};
\node[circle, draw, fill, minimum size = 5pt, inner sep = 0, outer sep = 0] (v6) at (2.41,0.8) {};

\draw[rotate=-50]  (0,-1.4) ellipse (0.8 and 1.4);
\node[circle, draw, fill, minimum size = 5pt, inner sep = 0, outer sep = 0] (v3) at (-0.6,-0.2) {};
\node[circle, draw, fill, minimum size = 5pt, inner sep = 0, outer sep = 0] (v4) at (-0.3,-0.6) {};

\draw[rotate=50]  (0.5,-2.0) ellipse (0.8 and 1.4);
\node[circle, draw, fill, minimum size = 5pt, inner sep = 0, outer sep = 0] (v2) at (1.6,-0.2) {};
\node[circle, draw, fill, minimum size = 5pt, inner sep = 0, outer sep = 0] (v1) at (1.2,-0.5) {};

\node[circle, draw, fill, minimum size = 5pt, inner sep = 0, outer sep = 0] (v7) at (5.2,-0.1) {};
\node[circle, draw, fill, minimum size = 5pt, inner sep = 0, outer sep = 0] (v8) at (4.6,-0.3) {};

\node at (-0.2,-0.2) {$S_1$};
\node at (1.1,-0.1) {$S_2$};
\node at (4.8,0.03) {$S_3$};

\begin{scope}
\begin{pgfinterruptboundingbox} 
  \path[invclip, scale =0.8]   (4.22,0.5) rectangle (5.46,-0.3);
  \end{pgfinterruptboundingbox}
\begin{scope}
\clip  (5,-1.7) ellipse (1.5 and 2);
\begin{scope}
  \begin{pgfinterruptboundingbox} 
  \path[invclip,rotate=-20, scale = 0.8]  (4.4,0.5) ellipse (0.6 and 1.4);
  \end{pgfinterruptboundingbox}
    \draw[rotate=20]   (5,-3) ellipse (0.6 and 1.4);
\end{scope}

\begin{scope}
  \begin{pgfinterruptboundingbox} 
  \path[invclip,rotate=20, scale = 0.8]   (5,-3) ellipse (0.6 and 1.4);
  \end{pgfinterruptboundingbox}
    \draw[rotate=-20]  (4.4,0.5) ellipse (0.6 and 1.4);
\end{scope}
\end{scope} 

\end{scope}
\begin{scope}
\clip   (4.22,0.5) rectangle (5.46,-0.3);
  \begin{pgfinterruptboundingbox} 
  \clip[rotate=-30]   (3,2.5) ellipse (2 and 1.1);
\draw  (5,-1.7) ellipse (1.5 and 2);
  \end{pgfinterruptboundingbox}
  \end{scope}   
\node at (4.2,-1.5) {$K$};
\node at (5.9,-1.5) {$K'$};

\draw[thick] (v1) -- (v2);

\begin{scope}[shift = {(9.5,2)}]

\draw[rotate=30]  (1,0) ellipse (2 and 1.3);

\draw[rotate=-30]  (3,2.5) ellipse (2 and 1.1);
\node[circle, draw, fill, minimum size = 5pt, inner sep = 0, outer sep = 0] (v5) at (2.34,1.4) {};
\node[circle, draw, fill, minimum size = 5pt, inner sep = 0, outer sep = 0] (v6) at (2.41,0.8) {};

\node[circle, draw, fill, minimum size = 5pt, inner sep = 0, outer sep = 0] (v3) at (-0.6,-0.2) {};
\node[circle, draw, fill, minimum size = 5pt, inner sep = 0, outer sep = 0] (v4) at (-0.3,-0.6) {};

\node[circle, draw, fill, minimum size = 5pt, inner sep = 0, outer sep = 0] (v2) at (1.6,-0.2) {};
\node[circle, draw, fill, minimum size = 5pt, inner sep = 0, outer sep = 0] (v1) at (1.2,-0.5) {};

\node[circle, draw, fill, minimum size = 5pt, inner sep = 0, outer sep = 0] (v7) at (5.2,-0.1) {};
\node[circle, draw, fill, minimum size = 5pt, inner sep = 0, outer sep = 0] (v8) at (4.6,-0.3) {};

\draw[thick] (v1) -- (v2);
\draw[thick] (v3) -- (v4);
\draw[thick] (v7) -- (v8);

\node at (-2,-0.5) {$G_{\bot}$};
\end{scope}

\begin{scope}[shift = {(9.5,-1)}]

\draw[rotate=-50]  (0,-1.4) ellipse (0.8 and 1.4);
\node[circle, draw, fill, minimum size = 5pt, inner sep = 0, outer sep = 0] (v3) at (-0.6,-0.2) {};
\node[circle, draw, fill, minimum size = 5pt, inner sep = 0, outer sep = 0] (v4) at (-0.3,-0.6) {};

\draw[rotate=50]  (0.5,-2.0) ellipse (0.8 and 1.4);
\node[circle, draw, fill, minimum size = 5pt, inner sep = 0, outer sep = 0] (v2) at (1.6,-0.2) {};
\node[circle, draw, fill, minimum size = 5pt, inner sep = 0, outer sep = 0] (v1) at (1.2,-0.5) {};

\node[circle, draw, fill, minimum size = 5pt, inner sep = 0, outer sep = 0] (v7) at (5.2,-0.1) {};
\node[circle, draw, fill, minimum size = 5pt, inner sep = 0, outer sep = 0] (v8) at (4.6,-0.3) {};

\node at (-0.2,-0.2) {$S_1$};
\node at (1.1,-0.1) {$S_2$};
\node at (4.8,0.03) {$S_3$};

\begin{scope}
\begin{pgfinterruptboundingbox} 
  \path[invclip,scale =0.8,shift = {(9.5,-1)}]   (4.22,0.5) rectangle (5.46,-0.3);
  \end{pgfinterruptboundingbox}
\begin{scope}
\begin{scope}
  \begin{pgfinterruptboundingbox} 
  \path[invclip,scale =0.8,shift = {(9.5,-1)},rotate=-20]  (4.4,0.5) ellipse (0.6 and 1.4);
  \end{pgfinterruptboundingbox}
    \draw[rotate=20]   (5,-3) ellipse (0.6 and 1.4);
\end{scope}

\begin{scope}
  \begin{pgfinterruptboundingbox} 
  \path[invclip,scale =0.8,shift = {(9.5,-1)},rotate=20]   (5,-3) ellipse (0.6 and 1.4);
  \end{pgfinterruptboundingbox}
    \draw[rotate=-20]  (4.4,0.5) ellipse (0.6 and 1.4);
\end{scope}
\end{scope} 

\end{scope}
\begin{scope}
\clip   (4.22,0.5) rectangle (5.46,-0.3);
  \begin{pgfinterruptboundingbox} 
  \clip[rotate=-30]   (3,2.5) ellipse (2 and 1.1);
\draw  (5,-1.7) ellipse (1.5 and 2);
  \end{pgfinterruptboundingbox}
  \end{scope}   
  
  \begin{scope}
  \clip[rotate=20]   (5,-3) ellipse (0.6 and 1.4);
\draw[dotted, rotate=-30]  (3,2.5) ellipse (2 and 1.1);
    \end{scope} 
      \begin{scope}
      \clip[rotate=-20]  (4.4,0.5) ellipse (0.6 and 1.4);
    \draw[dotted, rotate=-30]  (3,2.5) ellipse (2 and 1.1);
        \end{scope} 

\node at (4.2,-1.5) {$K$};

\node at (5.9,-1.5) {$K'$};

\draw[thick] (v1) -- (v2);
\draw[thick] (v3) -- (v4);
\draw[thick] (v7) -- (v8);

\end{scope}
\node (v9) at (7,0.5) {};
\node (v10) at (16,0.5) {};
\draw  (v9) edge (v10);
\node (v11) at (6.75,4) {};
\node (v12) at (6.75,-3.5) {};

\draw[dashed]  (v11) edge (v12);
\node at (8.5,-0.5) {$G_{\top}^{S_1}$};
\node at (15.5,-0.5) {$G_{\top}^{S_3}$};
\node at (11.5,-0.5) {$G_{\top}^{S_2}$};
\end{tikzpicture}

\caption{The figure illustrates the notions of the graphs~$G_{\bot}$ and~$G_{\top}^S$. In the~2-connected graph on the left the~2-separators are indicated and the respective decomposed graphs are shown on the right.}\label{fig:decomposed:graphs}
\end{figure}
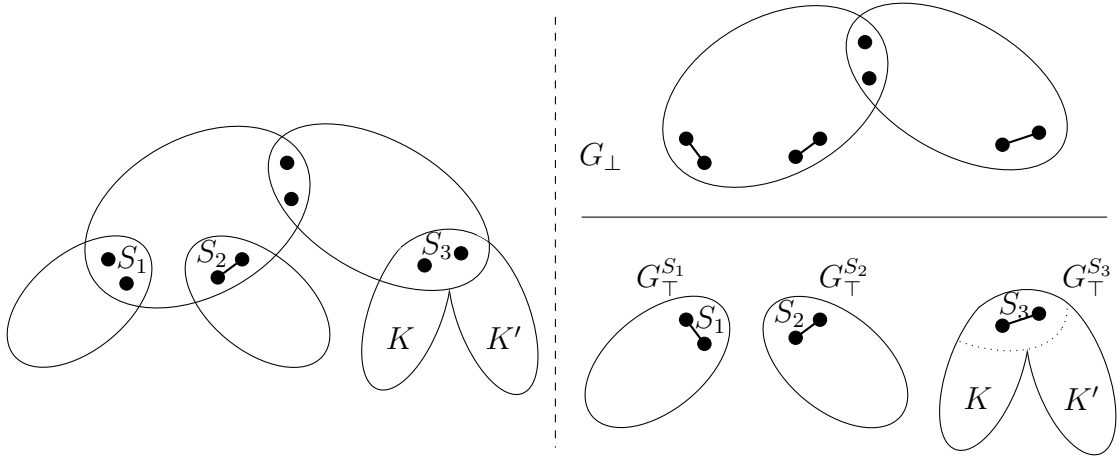

For an arc coloring~$\lambda$ of~$G$, we define an arc coloring~$\lambda^S_\top$ for the graph~$G^S_{\top}$ as follows:
\[\lambda^S_\top(v_1,v_2)  \coloneqq
\begin{cases}
(0,0)& \text{if } \{v_1, v_2\}\subseteq  S \text{ and } \{v_1, v_2\} \notin E(G)\\ 
(\lambda(v_1,v_2),1)& \text{if } \{v_1, v_2\}\subseteq  S \text{ and } \{v_1, v_2\} \in E(G) \\ 
(\lambda(v_1,v_2),2) & \text{otherwise.} 
\end{cases}
\]

If~$G$ is a vertex-colored graph with vertex coloring~$\lambda'$, in order to obtain a coloring for~$G^S_{\top}$, we define an arc coloring~$\lambda$ as~$\lambda (v_1,v_2) \coloneqq \lambda'(v_1)$ and let~$\lambda^S_\top$ be as above.

For~$(S,K) \in P_0(G)$ we also define~$G^{(S,K)}_{\top}\coloneqq G^S_{\top}[S \cup K]$, which differs from~$G^{S}_{\top}$ in that only vertices from~$S$ and~$K$ are retained. Again we define a coloring~$\lambda^{(S,K)}_\top$ which is simply the restriction of~$\lambda^S_\top$ to pairs~$(v_1,v_2)$ for which~$v_1,v_2\in S\cup K$.

Given a graph~$G$, we define~$G_{\bot}$ (see Figure~\ref{fig:decomposed:graphs} left and top right) to be the graph with vertex set~\[V_{\bot} \coloneqq V(G) \setminus \bigg(\bigcup_{(S,K)\in P_0(G)}K\bigg),\]
and edge set 
 \[E_{\bot} \coloneqq E(G[V_{\bot}]) \, \cup \, \{ \{s_1,s_2\} \mid \exists (S,K)\in P_0(G) \text{ s.t. } s_1,s_2\in S, s_1\neq s_2\}.\]
We observe that if~$G$ is not~2-connected then~$G_{\bot}$ is equal to~$G[V_{\bot}]$. In general, if for some~$k$ the graph~$G$ is not~$(k+1)$-connected but has minimum degree at least~$\frac{3k-1}{2}$, then Lemma~\ref{lem:minimum:degree:implies:disjoin:min:seps} applies. In particular, the various components whose vertex sets appear in~$P_0(G)$ are disjoint. If~$G$ is not~3-connected, this implies that~$G_{\bot}$ is a minor of~$G$.

In the following we restrict our discussions to graphs that are not~3-connected. Given an arc coloring~$\lambda$ of~$G$ we define an arc coloring~$\lambda_\bot$ of~$G_{\bot}$ as follows. Assume that~$v_1, v_2 \in V(G_\bot)$. Let~$S \coloneqq \{v_1,v_2\}$.

If~$S$ is a~$2$-separator of~$G$ but~$S \notin E(G)$, we set
\begin{align*}
\lambda_\bot(v_1,v_2) &\coloneqq 
\left(0, \ISOTYPE\left(\left(G^{S}_{\top},\lambda^{S}_\top\right)_{(v_1,v_2)}\right)\right) .
\intertext{Furthermore, if~$v_1 = v_2$ or if~$\{v_1,v_2\} \in E(G)$ we set}
\lambda_\bot(v_1,v_2) &\coloneqq \left(\lambda(v_1,v_2), \ISOTYPE\left(\left(G^S_{\top},\lambda^S_\top\right)_{(v_1,v_2)}\right)\right), 
\end{align*}
where by~$ \ISOTYPE((G^S_{\top},\lambda^S_\top)_{(v_1,v_2)})$  we denote the isomorphism class of the colored graph~$(G^S_{\top},\lambda^S_\top)_{(v_1,v_2)}$ obtained from the arc-colored graph~$(G^S,\lambda^S_\top)$ by individualizing~$v_1$ and~$v_2$. Thus~$(G^{\{v_1,v_2\}}_{\top},\lambda^{\{v_1,v_2\}}_\top)_{(v_1,v_2)}$ and~$(G'^{\{v'_1,v'_2\}}_{\top},\lambda'^{\{v'_1,v'_2\}}_\top)_{(v'_1,v'_2)}$ have the same isomorphism type if and only if there is an isomorphism from the first graph to the second mapping~$v_1$ to~$v'_1$ and~$v_2$ to~$v'_2$. Note that by definition, the~$\lambda_{\bot}$-colors of~$2$-separators of~$G$ are distinct from those of other pairs of vertices.

If not stated otherwise, we implicitly assume that for a graph~$G$ with initial coloring~$\lambda$, the corresponding graph~$G_\bot$ is a colored graph with initial coloring~$\lambda_\bot$.

\begin{lemma}\label{lem:iso:if:and:only:if:smaller:iso}
For~$k\in \{1,2\}$, if~$G$ and~$G'$ are~$k$-connected graphs that are not~$(k+1)$-connected and that are of minimum degree at least~$\frac{3k-1}{2}$ with arc colorings~$\lambda$ and~$\lambda'$, respectively, then
\[(G, \lambda) \cong (G', \lambda')  \Longleftrightarrow (G_{\bot}, \lambda_\bot) \cong (G'_{\bot},\lambda'_\bot).\] 
\end{lemma}

\begin{proof}
\rm($\Longrightarrow$) Suppose that~$\varphi$ is an isomorphism from~$(G, \lambda)$ to~$(G', \lambda')$. Since~$P_0(G)$ is isomorphism-invariant (Remark~\ref{remark:isomorphism_invariant}), we know that~$\varphi(V(G_{\bot}))= V(G'_{\bot})$. We claim that~$\varphi$ induces an isomorphism from~$(G_{\bot}, \lambda_\bot)$ to~$(G'_{\bot},\lambda'_\bot)$.
For this it suffices to observe that the definitions of~$G_{\bot}$ from~$G$ and~$\lambda_{\bot}$ from~$\lambda$ are isomorphism invariant.

\rm($\Longleftarrow$)
Conversely suppose~$\widetilde{\varphi}$ is an isomorphism from the graph~$(G_{\bot}, \lambda_\bot)$ to~$(G'_{\bot},\lambda'_\bot)$. Let~$\{S_1,\dots,S_t\}\coloneqq  \{S\mid \exists(S,K)\in P_0(G)\}$ be the set of separators that appear in~$P_0(G)$.
Since~$\widetilde{\varphi}$ respects the colorings~$\lambda_\bot$ and~$\lambda'_\bot$ we can conclude that
\[\{\widetilde{\varphi}(S_1),\dots,\widetilde{\varphi}(S_t)\}=  \{S\mid \exists(S,K)\in P_0(G')\}.\]
For each~$j\in \{1,\dots,t\}$ we choose an isomorphism~$\varphi_j$ from~$G^{S_j}_{\top}$ to~$G'^{\widetilde{\varphi}(S_j)}_{\top}$ that maps each~$s\in S_j$ to~$\widetilde{\varphi}(s)\in \widetilde{\varphi}(S_j)$. We know that such an isomorphism exists because~$\widetilde{\varphi}$ respects the colorings~$\lambda_\bot$ and~$\lambda'_\bot$.
We define a map~$\varphi$ from~$(G, \lambda)$ to~$(G', \lambda')$ by setting
\[\varphi(v)\coloneqq \begin{cases}
\widetilde{\varphi}(v) & \text{ if~$v \in V(G_{\bot})$} \\
\varphi_j(v) & \text{ if there is a set~$K\subseteq V(G)$ with~$v\in K$ and~$(S_j,K)\in P_0(G)$.}
\end{cases} \] 
This map is well-defined since by Parts~\ref{item:min:comp:trivial:int} and~\ref{item:char:min:comp} of Lemma~\ref{lem:minimum:degree:implies:disjoin:min:seps} the elements in the second components of pairs in~$P_0(G)$ are disjoint and not contained in~$V(G_{\bot})$. Moreover, the map is an isomorphism, since it respects all edges. Finally, by construction, it also respects the colors of vertices and arcs.
\end{proof}

\section{Reduction to vertex-colored~2-connected graphs}\label{sect:2connected}

It is easy to see that for a hereditary graph class~$\mathcal{G}$ and~$k \geq 2$, the~$k$-dimensional WL-algorithm distinguishes all (vertex-colored) graphs in~$\mathcal{G}$ if it distinguishes all (vertex-colored) connected graphs in~$\mathcal{G}$. (By a vertex-colored graph from~$\mathcal{G}$ we mean more precisely a colored graph whose underlying uncolored graph lies in~$\mathcal{G}$.) For this, one simply has to observe that for two non-isomorphic connected components, the sets of colors which the WL-algorithm computes for their vertices are disjoint. 

In this section we show a stronger statement replacing the assumption on connected graphs by an assumption on~2-connected graphs as follows. 

\begin{theorem}\label{thm:reduction2}
 Let~$\mathcal{G}$ be a hereditary graph class. If, for~$k \geq 2$, the~$k$-dimensional Weisfeiler-Leman algorithm distinguishes every two non-isomorphic~$2$-connected vertex-colored graphs~$(H, \lambda)$ and~$(H', \lambda')$ with~$H, H' \in \mathcal{G}$ from each other, then the~$k$-dimensional Weisfeiler-Leman algorithm distinguishes all non-isomorphic graphs in~$\mathcal{G}$.
\end{theorem}

For the rest of this section, let~$\mathcal{G}$ be a hereditary graph class. Recall that for a graph~$G \in \mathcal{G}$ with an initial vertex coloring or arc coloring~$\lambda$, the coloring~$\chi^k_G$ is the stable~$k$-tuple coloring produced by the~$k$-dimensional WL-algorithm on~$(G,\lambda)$.

For~$\ell$ vertices~$u_1, \dots, u_\ell$ with~$\ell < k$, we define
\[\chi^k_G (u_1, \dots, u_\ell) \coloneqq \chi^k_G (u_1, \dots, u_\ell, \underbrace{u_\ell, \dots, u_\ell}_{k-\ell \text{ times}})\]
to be the coloring of the~$k$-tuple resulting from extending the~$\ell$-tuple by repeating~$k-\ell$ times its last entry.

To prove Theorem~\ref{thm:reduction2}, we first show that the~$2$-dimensional WL-algorithm distinguishes pairs of vertices that lie in a common~$2$-connected component from pairs that do not.

\begin{theorem}\label{thm:edges_conncomp}
Assume~$k \geq 2$ and let~$G, H$ be two graphs. 
Let~$u$ and~$v$ be vertices from the same~$2$-connected component of~$G$ and let~$u'$ and~$v'$ be vertices that are not contained in a common~$2$-connected component of~$H$. Then~$\chi^k_G(u,v) \neq \chi^k_H(u',v')$.
\end{theorem}

\begin{proof}
To improve readability, in this proof we omit the superscripts~$k$, i.e., we write~$\chi_G$ and~$\chi_H$ instead of~$\chi^k_G$ and~$\chi^k_H$, respectively. 

For an integer~$i$ and vertices~$x,y$ denote by~$W_i(x,y)$ the number of walks of length exactly~$i$ from~$x$ to~$y$. (It will be clear from context in which graph we count the number of walks.)
By induction on~$i$, it is easy to see that for~$k\geq 2$ it holds that~$W_i(x,y)\neq W_i(x',y')$ implies~$\chi_G(x,y)\neq \chi_H(x',y')$ (\cite[p.~18]{MR0543783}).
  Thus, it suffices to show that for some~$i$, we have~$W_i(u,v) \neq W_i(u',v')$. Since~$u'$ and~$v'$ are not contained in the same~$2$-connected component, there is some cut vertex~$w'$ such that every walk from~$u'$ to~$v'$ passes~$w'$. Suppose that there does not exist a vertex~$w$ such that for all~$i$ the following hold: 
 \begin{enumerate}
  \item~$W_i(u,w) = W_i(u',w')$
  \item~$W_i(w,w) = W_i(w',w')$
  \item~$W_i(w,v) = W_i(w',v')$.
 \end{enumerate}

Then for every vertex~$w$ it holds that~$\chi_G(u,w) \neq \chi_H(u',w')$ or~$\chi_G(w,w) \neq \chi_H(w',w')$ or~$\chi_G(w,v) \neq \chi_H(w',v')$. If~$\chi_G(w) \neq \chi_H(w')$, then~$\chi_G(u,w) \neq \chi_H(u',w')$ and thus, for every vertex~$w$ it holds that~$\chi_G(u,w) \neq \chi_H(u',w')$ or~$\chi_G(w,v) \neq \chi_H(w',v')$. 
In other words, there is no vertex~$w$ such that~$( \chi_G(w,v),\chi_G(u,w)) = (\chi_H(w',v'),\chi_H(u',w'))$. By the definition of the WL-algorithm this implies that~$\chi_G(u,v) \neq \chi_H(u',v')$.  
 
 Now suppose that there is a vertex~$w$ such that for all~$i$ Conditions~1,~2 and~3 hold. Then for every~$i$ the number of walks of length~$i$ from~$u$ to~$v$ which pass~$w$ equals the number of walks from~$u'$ to~$v'$ which pass~$w'$. However, there must be a walk from~$u$ to~$v$ which avoids~$w$. Let~$d$ be its length. We have~$W_d(u,v) > W_d(u',v')$ and thus~$\chi_G(u,v) \neq \chi_H(u',v')$.
\end{proof}

Next we argue that for~$k \geq 2$, the~$k$-dimensional WL-algorithm distinguishes cut vertices from other vertices.

\begin{corollary}\label{cor:cut:vertices:are:different}
 Let~$k \geq 2$ and assume~$G, H$
 are connected graphs.
 Let~$w\in V(G)$ and~$w'\in V(H)$ be vertices such that~$G - \{w\}$ is connected and~$H - \{w'\}$ is disconnected. Then~$\chi^k_G(w) \neq \chi^k_H(w')$.
\end{corollary}

\begin{proof}
Let~$u'$ and~$v'$ be two neighbors of~$w'$ not sharing a common~$2$-connected component in~$H$. Note that such vertices do not exist for~$w$ in~$G$. 

It suffices to show that for all~$u \in V(G)$ it holds that~$\chi^k_G(u,w) \neq \chi^k_H(u',w')$. By Theorem~\ref{thm:edges_conncomp}, the color~$\chi^k_H(u',w')$ encodes that there is~$v'$ which is a neighbor of~$w'$ and which is not contained in the same~$2$-connected component as~$u'$. For~$u$ and~$w$, such a vertex does not exist in~$G$.
\end{proof}

We prove Theorem~\ref{thm:reduction2} by induction over the sizes of the input graphs. The strategy is to show that on input~$(G, \lambda)$ the WL-algorithm implicitly computes the graph~$(G_\perp, \lambda_\perp)$ and to then apply Lemma~\ref{lem:iso:if:and:only:if:smaller:iso}.

\begin{lemma}\label{lem:bot:vertices:are:different}
	Let~$k \geq 2$ and assume~$G, H$ are connected graphs that are not~2-connected. For vertices~$v \in V(G_\perp)$ and~$w \in V(H) \setminus V(H_\perp)$ we have~$\chi^k_G(v) \neq \chi^k_H(w)$.
\end{lemma}

\begin{proof}
Note that for a connected but not~2-connected graph~$G$, a vertex~$v\in V(G)$ is in~$V(G_\perp)$ if and only if it is a cut vertex or there are at least two cut vertices that lie in the same~2-connected component as~$v$. The equivalent statement holds for~$H$. If~$v$ is a cut vertex of~$G$, then the lemma follows immediately from Corollary~\ref{cor:cut:vertices:are:different}.

If~$v$ is not a cut vertex, then there are at least two cut vertices~$u$ and~$u'$ lying in the same~2-connected component as~$v$. Note that there are no such two vertices for~$w$. By Corollary~\ref{cor:cut:vertices:are:different},~$u$ and~$u'$ obtain colors distinct from colors of non-cut vertices. Thus, also the colors~$\chi^k_G(v,u)$ and~$\chi^k_G(v,u')$ are distinct from all colors of edges from~$v$ to non-cut vertices. Moreover, Theorem~\ref{thm:edges_conncomp} yields that the colors~$\chi^k_G(v,u)$ and~$\chi^k_G(v,u')$ also encode that~$v,u,u'$ all share a common~2-connected component. This information about the existence of such~$u$ and~$u'$ is contained in the color~$\chi^k_G(v)$ and thus,~$\chi^k_G(v) \neq \chi^k_H(w)$.
\end{proof}

\begin{lemma}\label{lem:one:component:hanging}
For graphs~$G, G' \in \mathcal{G}$ with vertex colorings~$\lambda$ and ~$\lambda'$, respectively, assume~$(s,K) \in P_0(G)$ and~$(s',K') \in P_0(G')$. For~$k \geq 2$ suppose the~$k$-dimensional Weisfeiler-Leman algorithm distinguishes all non-isomorphic vertex-colored~2-connected graphs in~$\mathcal{G}$. Assume there is no isomorphism from~$(G_{\top}^{(s,K)}, \lambda_\top^{(s,K)})$ to~$(G_{\top}'^{(s',K')}, \lambda_\top'^{(s',K')})$ that maps~$s$ to~$s'$. Then
\[
\{ \chi^k_G(s,v) \mid v \in K\} \cap \{ \chi^k_{G'}(s',v') \mid v' \in K'\} = \varnothing.
\]
\end{lemma}

\begin{proof}
If~$\chi^k_G(s)\neq \chi^k_{G'}(s')$ then the conclusion of the lemma is obvious. Thus, we can assume otherwise.
We have already seen with Corollary~\ref{cor:cut:vertices:are:different} that cut vertices obtain different colors than non-cut vertices. Thus, we can assume that~$G$ and~$G'$ are already colored in a way such that~$s$ and~$s'$  have a color different from the colors of vertices in~$K \cup K'$. 

With Theorem~\ref{thm:edges_conncomp} we will now argue that for~$v\in K$ and~$v'\in K'$ we have~$\chi^k_G(v)\neq \chi^k_{G'}(v')$, which implies the lemma. For readability, we drop the superscripts~$(s,K)$ and~$(s',K')$. 
We will show by induction that if the lemma does not hold, then for all~$u,v \in K \cup \{s\}$ and all~$u',v' \in K' \cup \{s'\}$ with~$\{u,v\} \not\subseteq \{s\}$ and~$\{u',v'\} \not\subseteq \{s'\}$ the following implication is true: 
\begin{align}\label{imp:induction}
\prescript{}{i} \chi_{G_\top}^k(u,v) \neq \prescript{}{i} \chi_{G'_\top}^k(u',v') \Rightarrow  \prescript{}{i} \chi_G^k(u,v) \neq \prescript{}{i} \chi_{G'}^k(u',v').
\end{align}

For~$i = 0$ the claim follows by definition of the colorings~$\lambda_\top$ and~$\lambda'_\top$. For the induction step, assume that there exist vertices~$x,y \in K \cup \{s\}$,~$x',y' \in K' \cup \{s'\}$ such that~$\{x,y\} \not\subseteq \{s\}$,~$\{x',y'\} \not\subseteq \{s'\}$  with~$\prescript{}{i} \chi_{G_\top}^{k}(x,y) = \prescript{}{i} \chi_{G'_\top}^{k}(x',y')$ and~$\prescript{}{i+1} \chi_{G_\top}^{k}(x,y) \neq \prescript{}{i+1} \chi_{G'_\top}^{k}(x',y')$. Thus, there must be a color tuple~$(c_1,c_2)$ such that the sets
\begin{align*}
M \coloneqq{} & \big\{w \mid w \in V(G_\top) \backslash \{x,y\}, (\prescript{}{i} \chi^k_{G_\top}(w,y), 
\prescript{}{i} \chi^k_{G_\top}(x,w)) = (c_1,c_2)\big\}
\intertext{and} 
M' \coloneqq{} & \big\{w' \mid w' \in V(G'_\top) \backslash \{x',y'\}, (\prescript{}{i} \chi^k_{G'_\top}(w',y'), \prescript{}{i} \chi^k_{G'_\top}(x',w')) = (c_1,c_2)\big\}
\intertext{do not have the same cardinality. Let}
D \coloneqq{} & \{(\prescript{}{i} \chi^k_{G}(w,y), \prescript{}{i} \chi^k_{G}(x,w)) \mid w \in M\} \cup \{(\prescript{}{i} \chi^k_{G'}(w',y'), \prescript{}{i} \chi^k_{G'}(x',w')) \mid w' \in M'\}.
\end{align*} 
By induction and by Theorem~\ref{thm:edges_conncomp} 
we have that
\begin{align*}
\big\{w \mid w \in V(G) \backslash \{x,y\}, (\prescript{}{i} \chi^k_{G}(w,y), \prescript{}{i} \chi^k_{G}(x,w)) \in D\big\} &= M
\intertext{and}
\big\{w' \mid w' \in V(G') \backslash \{x',y'\}, (\prescript{}{i} \chi^k_{G'}(w',y'), \prescript{}{i} \chi^k_{G'}(x',w')) \in D\big\} &= M'
\end{align*}
and hence these sets do not have the same cardinality. Thus,~$\prescript{}{i+1} \chi_G^{k}(x,y) \neq \prescript{}{i+1} \chi_{G'}^{k}(x',y')$. 

Having shown Implication~(\ref{imp:induction}), it suffices to show that 
\[
\{ \chi^k_{G_\top}(s,v) \mid v \in K\} \cap \{ \chi^k_{G'_\top}(s',v') \mid v' \in K'\} = \varnothing.
\]
However, this follows directly from the assumption that the~$k$-dimensional WL-algorithm distinguishes every pair of non-isomorphic vertex-colored~2-connected graphs in~$\mathcal{G}$ and that the graphs~$(G_{\top}^{(s,K)}, \lambda_\top^{(s,K)})$ and~$(G_{\top}'^{(s',K')}, \lambda_\top'^{(s',K')})$ are~2-connected.
\end{proof}

With this, we can prove the following.
	
\begin{lemma}\label{lem:distinct_cutvertices}
Assume~$k \geq 2$ and suppose the~$k$-dimensional Weisfeiler-Leman algorithm distinguishes all non-isomorphic vertex-colored~2-connected graphs in~$\mathcal{G}$. For two graphs~$G, G' \in \mathcal{G}$ with vertex colorings~$\lambda, \lambda'$, respectively, suppose~$s \in V(G), s' \in V(G')$. Assume there is no isomorphism from~$(G_{\top}^{s}, \lambda_\top^{s})$ to~$(G_{\top}'^{s'}, \lambda_\top'^{s'})$ that maps~$s$ to~$s'$. 

Then~$\chi^k_G(s) \neq \chi^k_{G'}(s')$.
\end{lemma}

\begin{proof}
Assume otherwise that~$\chi^k_G(s) = \chi^k_{G'}(s')$.
Further suppose that
\begin{align*}
\{K_1,\dots,K_t\} & = \{K \mid (s,K) \in P_0(G)\}
\intertext{and that}
\{K'_1,\dots,K'_{t'}\} & = \{K' \mid (s',K') \in P_0(G)\}.
\end{align*}
From~$(G_{\top}^{s}, \lambda_\top^{s}) \not\cong (G_{\top}'^{s'}, \lambda_\top'^{s'})$ we conclude that there is a vertex-colored graph~$(H,\lambda_H)$ such that the sets
\begin{align*}
I & \coloneqq \left\{j\mid \left(G_{\top}^{(s,K_j)}, \lambda_\top^{(s,K_j)}\right) \cong (H,\lambda_H)\right\}
\intertext{and}
I' & \coloneqq \left\{j\mid \left({G'}_{\top}^{(s',K'_j)}, \lambda_\top'^{(s',K'_j)}\right) \cong (H,\lambda_H)\right\}
\end{align*}
have different cardinalities.
Note that all~$K_j$ with~$j\in I$ and all~$K'_j$ with~$j\in I'$ have the same cardinality.
We know by Lemma~\ref{lem:one:component:hanging} that for~$v\in K_i$ with~$i\in I$ and~$v'\in K_j$ with~$j\notin I'$ we have~$\chi^k_G(s,v)\neq \chi^k_{G'}(s',v')$. Letting~$C \coloneqq \{\chi^k_G(s,v)\mid i\in I \text{ and } v\in K_i\}$, the vertices~$s$ and~$s'$ do not have the same number of neighbors connected via an arc of color~$C$.
We conclude that~$\chi^k_G(s) \neq \chi^k_{G'}(s')$.
\end{proof}

\begin{corollary}\label{cor:induced:on:bot:graphs:is:finer}
	Assume~$k \geq 2$ and suppose the~$k$-dimensional Weisfeiler-Leman algorithm distinguishes all non-isomorphic vertex-colored~2-connected graphs in~$\mathcal{G}$.
	Let~$G, G' \in \mathcal{G}$ be connected graphs that are not~$2$-connected with vertex colorings~$\lambda, \lambda'$, respectively. If for vertices~$v_1,v_2 \in V(G_\perp)$ and~$v'_1,v'_2\in V(G'_\perp)$ we have~$\chi^k_{G_\perp}(v_1,v_2) \neq \chi^k_{G'_\perp}(v'_1,v'_2)$, then~$\chi^k_{G}(v_1,v_2) \neq \chi^k_{G'}(v'_1,v'_2)$.
\end{corollary}
\begin{proof}
By Lemma~\ref{lem:bot:vertices:are:different}, with respect to the colorings~$\chi^k_G$ and~$\chi^k_{G'}$, the vertices in~$V(G_{\bot})$  and~$V(G'_{\bot})$ have different colors than the vertices in~$V(G) \setminus V(G_{\bot})$ and~$V(G') \setminus V(G'_{\bot})$. Thus, it suffices to show that the colorings~$\chi^k_G$ and~$\chi^k_{G'}$ refine the colorings~$\lambda_\bot$ and~$\lambda'_\bot$, respectively. For this, by the definition of~$\lambda_\bot$ and~$\lambda'_\bot$, it suffices to show the following two statements.

\begin{enumerate}
  \item If we have that~$v_1=v_2$ and~$v'_1=v'_2$ and also~$\ISOTYPE((G^{\{v_1\} }_{\top},\lambda^{\{v_1\}}_\top)_{(v_1)}) \neq
\ISOTYPE((G'^{\{v'_1\} }_{\top},\lambda'^{\{v'_1\}}_\top)_{(v'_1)})$, then~$\chi^k_{G}(v_1) \neq \chi^k_{G'}(v'_1)$. 
  \item If we have~$\{v_1,v_2\}\in E(G)$ and~$\{v'_1,v'_2\}\in E(G')$ and also~$\lambda_\bot(v_1,v_2) \neq \lambda'_\bot(v'_1,v'_2)$, then~$\chi^k_{G}(v_1,v_2) \neq \chi^k_{G'}(v'_1,v'_2)$.
\end{enumerate}

For the first item, from~$\ISOTYPE((G^{\{v_1\} }_{\top},\lambda^{\{v_1\}}_\top)_{(v_1)}) \neq
\ISOTYPE((G'^{\{v'_1\} }_{\top},\lambda'^{\{v'_1\}}_\top)_{(v'_1)})$ we know that~$v_1$ and~$v'_1$ must be cut vertices. Thus, the statement is exactly Lemma~\ref{lem:distinct_cutvertices}. For the second item, from the definition of~$\lambda_\bot$ and~$\lambda'_\bot$ we obtain~$\lambda(v_1,v_2) \neq \lambda'(v'_1,v'_2)$, which implies~$\chi^k_{G}(v_1,v_2) \neq \chi^k_{G'}(v'_1,v'_2)$. 
\end{proof}

\begin{proof}[Proof of Theorem~\ref{thm:reduction2}] Let~$(G, \lambda)$ and~$(G', \lambda')$ be vertex-colored graphs in~$\mathcal{C}$.
We prove the statement by induction on~$|V(G)|+|V(G')|$. If both graphs are~2-connected, then the statement follows directly from the assumptions. If exactly one of the graphs is~2-connected, then exactly one of the graphs has a cut vertex and the statement follows from Lemma~\ref{cor:cut:vertices:are:different}.
Thus suppose both graphs are not~2-connected but connected.
Since~$(G, \lambda) \not\cong (G', \lambda')$, we know by Lemma~\ref{lem:iso:if:and:only:if:smaller:iso} that~$(G_{\bot}, \lambda_\bot) \not\cong (G'_{\bot},\lambda'_\bot)$.
By Lemma~\ref{lem:bot:vertices:are:different} the vertices in~$V(G_{\bot})$  and~$V(G'_{\bot})$ have different colors than the vertices in~$V(G) \setminus V(G_{\bot})$ and~$V(G') \setminus V(G'_{\bot})$.
Moreover by Corollary~\ref{cor:induced:on:bot:graphs:is:finer}, the partition of the vertices and arcs induced by the coloring~$\chi^k_G$ restricted to~$V(G_{\bot})$ is finer than the partition induced by~$\lambda_\bot$. Similarly, the partition induced by~$\chi^k_{G'}$ on~$V(G'_{\bot})$  is finer than the partition induced by~$\lambda'_\bot$. By induction, the~$k$-dimensional WL-algorithm distinguishes~$(G_{\bot}, \lambda_\bot)$ from~$(G'_{\bot},\lambda'_\bot)$. Thus, the~$k$-dimensional WL-algorithm distinguishes~$(G, \lambda)$ from~$(G', \lambda')$.
\end{proof}

\section{Reduction to arc-colored~3-connected graphs}\label{sec:reduct:to:3:con}

In this section our aim is to weaken the assumption from Theorem~\ref{thm:reduction2} which requires that~2-connected graphs are distinguished to an assumption of~3-connected graphs being distinguished.

The strategy to prove our reduction follows similar ideas as those used in Section~\ref{sect:2connected}. It relies on the assumption that the input consists of vertex-colored~$2$-connected graphs, which we can make without loss of generality by the reduction from the last section. Now we consider the decomposition of vertex-colored~$2$-connected graphs into their so-called~``$3$-connected components''. 

Most of the results stated in Section~\ref{sect:2connected} have analogous formulations for the~$3$- or higher-dimensional WL-algorithm on~$2$-connected graphs. But a~$3$-connected component of a~$2$-connected graph~$G$ is not necessarily a subgraph and may only be a minor of~$G$. Thus, we require that the graph class~$\mathcal{G}$ is minor-closed. Furthermore, to enable the inductive approach we will now have to consider graphs~$G$ in which the~2-tuples~$(u,v)$ with~$\{u,v\} \in E(G)$, i.e., the arcs, are also colored. However, it turns out that it is not sufficient to require that arc-colored graphs are distinguished. In fact we need the following stronger property.

\begin{definition}
 Let~$\mathcal{H}$ be a set of graphs. We say that the~$k$-dimensional WL-algorithm \emph{correctly determines orbits} in~$\mathcal{H}$ if for all arc-colored graphs~$(G, \lambda), (G',\lambda')$ with~$G, G' \in \mathcal{H}$ and all vertices~$s \in V(G)$ and~$s' \in V(G')$ the following holds: there exists an isomorphism from~$(G,\lambda)$ to~$(G',\lambda')$ mapping~$s$ to~$s'$ if and only if~$\chi^k_G(s) = \chi^k_{G'}(s')$.

Note that for~$G' = G$, the vertex color classes obtained by an application of the~$k$-dimensional WL-algorithm to an arc-colored graph~$(G, \lambda)$ are the orbits of the automorphism group of~$G$ with respect to~$\lambda$.
\end{definition}

The main result in this section is the following reduction theorem.

\begin{theorem}\label{thm:reduction3}
 Let~$\mathcal{G}$ be a minor-closed graph class and assume~$k \geq 3$. Suppose the~$k$-dimensional Weisfeiler-Leman algorithm correctly determines orbits on all arc-colored~$3$-connected graphs in~$\mathcal{G}$. Then the~$k$-dimensional Weisfeiler-Leman algorithm distinguishes all non-isomorphic graphs in~$\mathcal{G}$.
\end{theorem}

The next corollary states that the~$3$-dimensional WL-algorithm distinguishes~2-separators from other pairs of vertices.

\begin{corollary}\label{cor:separating_pairs}
 Assume~$k \geq 3$ and let~$G$ and~$H$ be~$2$-connected graphs. Let~$u,v,u',v'$ be vertices such that~$G - \{u,v\}$ is disconnected and~$H - \{u',v'\}$ is connected. Then~$\chi^k_G(u,v) \neq \chi^k_H(u',v')$. 
\end{corollary}

\begin{proof}
Consider the connected graphs~$G-\{u\}$ and~$H-\{u'\}$. In the first graph~$v$ is a cut vertex but in the second graph~$v'$ is not a cut vertex. Thus, by Corollary~\ref{cor:cut:vertices:are:different}, we have that~$\chi^{k-1}_{G-\{u\}}(v) \neq \chi^{k-1}_{H-\{u'\}}(v')$ and thus~$\chi^k_G(u,v) \neq \chi^k_H(u',v')$. 
\end{proof}

Just as we did in the previous section, we want to apply a recursive strategy that relies on Lemma~\ref{lem:iso:if:and:only:if:smaller:iso}.
However, to apply that lemma we require a minimum degree of~$3$. The following lemma states that vertices of degree~$2$ can be removed.

\begin{lemma}\label{lem:min:degree:3}
 Let~$\mathcal{G}$ be a minor-closed graph class and assume~$k \geq 2$. Suppose the~$k$-dimensional Weisfeiler-Leman algorithm correctly determines orbits on all arc-colored graphs in~$\mathcal{G}$ of minimum degree at least~3. Then the~$k$-dimensional Weisfeiler-Leman algorithm distinguishes all non-isomorphic arc-colored graphs in~$\mathcal{G}$.
\end{lemma}

\begin{proof}[Sketch of the proof] The proof is a basic exercise regarding the WL-algorithm. We give a sketch. By Theorem~\ref{thm:reduction2} we can assume that the graphs are~2-connected.
We first observe that the~2-dimensional WL-algorithm identifies graphs of maximum degree at most~2. Since vertices of degree at least~3 obtain different colors than vertices of degree at most~2 it suffices now to observe that for each~$i$ the color~$\chi^k_G(u,v)$ implicitly encodes the number of paths of length exactly~$i$ from~$u$ to~$v$ whose inner vertices are of degree~2. 
Inductively we can then consider the minors obtained by retaining vertices of degree at least~3 and connecting two such vertices with an edge if there is a path between them whose inner vertices all have degree~2. The edge is colored with a color that encodes the multiset of lengths of paths between the two vertices only having inner vertices of degree~2.
\end{proof}

The lemma allows us to focus on graphs with minimum degree~3. Doing so, in analogy to Lemma~\ref{lem:bot:vertices:are:different}, the following proposition gives a characterization of the vertices in~$V(G_\perp)$.

\begin{proposition}\label{prop:tool}
Assume~$k \geq 3$, let~$G$ be a~2-connected graph of minimum degree at least~3 that is not~$3$-connected. 
Then~$x \notin V(G_\perp)$ if and only if there exists a vertex~$u$ contained in some minimal separator of~$G$ such that~$x\notin V((G-\{u\})_\perp)$ and such that the (unique)~2-connected component containing~$x$ in~$G-\{u\}$ has exactly one vertex belonging to a minimal separator of~$G$.
\end{proposition}

\begin{proof}
If~$x\notin V(G_\perp)$ then~$x\in K$ for a~$(\{u,s\},K) \in P_0(G)$. Then~$(s,K) \in P_0(G-\{u\})$. Moreover, by Part~\ref{item:char:min:comp} of Lemma~\ref{lem:minimum:degree:implies:disjoin:min:seps}, no vertex of~$K$ is contained in a minimum separator of~$G$. This implies that~$s$ is the only vertex in the~2-connected component of~$x$ in~$G-\{u\}$ that belongs to a minimum separator of~$G$.

Conversely suppose~$u$ is a vertex in a minimal separator of~$G$ such that~$x\notin V((G-\{u\})_\perp)$ and the~2-connected component of~$x$ in~$G-\{u\}$ has exactly one vertex belonging to a minimal separator of~$G$.
Since~$x\notin V((G-\{u\})_\perp)$, there is~$(s,K)\in P_0(G-\{u\})$ with~$x\in K$. Then~$K\cup \{s\}$ is a~2-connected component of~$G-\{u\}$ and, since~$s$ is the only 
vertex in this~2-connected component that is contained in a minimal separator of~$G$,
by Part~\ref{item:char:min:comp} of Lemma~\ref{lem:minimum:degree:implies:disjoin:min:seps}, we have~$(\{u,s\},K) \in P_0(G)$.
\end{proof}

\begin{lemma}\label{lem:bot:vertices:are:different:3con}
Assume~$k \geq 3$, let~$G, H$ be~2-connected graphs of minimum degree at least~3 that are not~$3$-connected. 
Then for vertices~$v \in V(G_\perp)$ and~$w \in V(H) \setminus V(H_\perp)$ we have~$\chi^k_G(v) \neq \chi^k_H(w)$.
\end{lemma}

 \begin{proof}
Suppose that~$v \in V(G_\perp)$ and~$w \in V(H) \setminus V(H_\perp)$. Let~$u$ be a vertex contained in some minimal separator of~$H$ such that~$w\notin V((H-\{u\})_\perp)$ and such that the~2-connected component of~$w$ in~$H-\{u\}$ has exactly one vertex that is contained in a minimal separator of~$H$. Such a vertex exists by Proposition~\ref{prop:tool}. We argue that~$\chi^k_G(v,t)\neq \chi^k_H(w,u)$ for all~$t\in V(G)$. 

If~$t$ is not contained in any minimal separator, then this follows from Corollary~\ref{cor:separating_pairs}. Otherwise we know that~$v\in V((G-\{t\})_\perp)$ or the~2-connected component of~$v$ in~$G-\{t\}$ does not have exactly one vertex that is contained in a minimal separator of~$G$.
In the first case we use Lemma~\ref{lem:bot:vertices:are:different} and in the second case we use Theorem~\ref{thm:edges_conncomp} and Corollary~\ref{cor:separating_pairs} to conclude that~$\chi^{k-1}_{G-\{t\}}(v)\neq \chi^{k-1}_{H-\{u\}}(w)$ and thus~$\chi^k_G(v,t)\neq \chi^k_H(w,u)$.
\end{proof}

\begin{lemma}\label{lem:one:component:hanging:sep:pairs}
For~$k \geq 3$ suppose the~$k$-dimensional Weisfeiler-Leman algorithm correctly determines orbits on all arc-colored~$3$-connected graphs in~$\mathcal{G}$. Suppose~$G, G' \in \mathcal{G}$ are arc-colored~$2$-connected graphs of minimum degree at least~3. Assume that~$(\{s_1,s_2\},K) \in P_0(G)$ and~$(\{s'_1, s'_2\},K') \in P_0(G')$.

If no isomorphism from~$(G_{\top}^{(\{s_1, s_2\},K)}, \lambda_\top^{(\{s_1, s_2\},K)})$ to~$(G_{\top}'^{(\{s'_1, s'_2\},K')},\lambda_\top'^{(\{s'_1, s'_2\},K')})$ maps~$s_1$ to~$s'_1$ and~$s_2$ to~$s'_2$, then
\[
\{ \chi^k_G(s_1,s_2,v) \! \mid \! v \in K\} \cap \{ \chi^k_{G'}(s'_1,s'_2,v) \! \mid \! v \in K'\} = \varnothing.
\]
\end{lemma}

\begin{proof}
This is an adaption of the proof of Lemma~\ref{lem:one:component:hanging}.

If~$\chi^k_G(s_1,s_2) \neq \chi^k_{G'}(s'_1,s'_2)$ then the conclusion of the lemma is obvious. Thus, we can assume otherwise. We have already seen with Corollary~\ref{cor:separating_pairs} that~2-separators obtain different colors than other pairs of vertices. Thus, we can assume that~$G$ and~$G'$ are already colored in a way such that~$(s_1, s_2)$,~$(s_2, s_1)$,~$(s'_1, s'_2)$,~$(s'_2, s'_1)$ have colors different from the colors of pairs of vertices~$(t_1,t_2)$ with~$\{t_1, t_2\} \cap (K \cup K') \neq \varnothing$. This immediately implies that we can assume that~$s_1, s_2, s'_1, s'_2$ have colors different from colors of vertices that are not contained in any~$2$-separator in the graphs.

We will now argue that for~$v \in K$ and~$v' \in K'$ we have~$\chi^k_G(v) \neq \chi^k_{G'}(v')$, which implies the lemma. For readability, we drop the superscripts~$(\{s_1, s_2\},K)$ and~$(\{s'_1, s'_2\},K')$. 

We will show by induction that if the lemma does not hold, then for all~$i \in \mathbb{N}$, all~$u,v,w \in K \cup \{s_1, s_2\}$ and all~$u',v',w' \in K' \cup \{s'_1, s'_2\}$ with~$\{u,v,w\} \not\subseteq \{s_1,s_2\}$ and~$\{u',v',w'\} \not\subseteq \{s'_1,s'_2\}$ the following implication holds: 
\begin{align}\label{imp:induction2}
\prescript{}{i} \chi_{G_\top}^k(u,v,w) \neq \prescript{}{i} \chi_{G'_\top}^k(u',v',w') \Rightarrow  \prescript{}{i} \chi_G^k(u,v,w) \neq \prescript{}{i} \chi_{G'}^k(u',v',w').
\end{align}

For the induction base with~$i = 0$, suppose~$\prescript{}{0} \chi_{G_\top}^k(u,v,w) \neq \prescript{}{0} \chi_{G'_\top}^k(u',v',w')$. Then either there is no isomorphism from~$G_\top[u,v]$ to~$G'_\top[u',v']$ mapping~$u$ to~$u'$ and~$v$ to~$v'$, or~$\lambda_\top(u,v,w) \neq \lambda'_\top(u',v',w')$. 

In the first case, by definition of the graphs~$G_\top$ and~$G'_\top$ and their colorings, we immediately get~$\prescript{}{0} \chi_G^k(u,v,w) \neq \prescript{}{0} \chi_{G'}^k(u',v',w')$ since an isomorphism from~$G$ to~$G'$ that maps~$u$ to~$u'$ and~$v$ to~$v'$ would induce an isomorphism from~$G_\top$ to~$G'_\top$ with the same mappings. In the second case, since~$G$ and~$G'$ are arc-colored graphs, we have~$\lambda_\top(u,v,w) = \lambda_\top(u,v)  \neq  \lambda'_\top(u',v') = \lambda'_\top(u',v',w')$. If~$\lambda(u,v)  \neq  \lambda'(u',v')$ the{}n~$\prescript{}{0} \chi_G^k(u,v,w) \neq \prescript{}{0} \chi_{G'}^k(u',v',w')$. Otherwise, from the definitions of~$\lambda_\top$ and~$\lambda'_\top$ we conclude that~$(u,v)\in E(G) \nLeftrightarrow (u',v')\in E(G')$ or
$\{u,v\}\subseteq \{s_1, s_2\}  \nLeftrightarrow \{u',v'\}\subseteq  \{s'_1, s'_2\}$. In the first subcase, we conclude that~$\prescript{}{0} \chi_G^k(u,v,w) \neq \prescript{}{0} \chi_{G'}^k(u',v',w')$. In the second subcase, since~$(\{s_1,s_2\},K) \in P_0(G)$ and~$(\{s'_1,s'_2\},K') \in P_0(G')$ we know with Part~\ref{item:one:min:comp:trivial:int:contained} from Lemma~\ref{lem:minimum:degree:implies:disjoin:min:seps} that either at least one of~$u$ and~$v$ or at least one of~$u'$ and~$v'$ is a vertex not contained in any~$2$-separator. Since we have assumed that~$s_1, s_2, s'_1, s'_2$ have colors that are different from colors of vertices not contained in any~$2$-separator, we also conclude that~$\prescript{}{0} \chi_G^k(u,v,w) \neq \prescript{}{0} \chi_{G'}^k(u',v',w')$. 

For the induction step, assume that there exist vertices~$x,y,z \in K \cup \{s_1,s_2\}$ and vertices~$x',y',z' \in K' \cup \{s'_1, s'_2\}$ such that~$\{x,y,z\} \not\subseteq \{s_1, s_2\}$ and~$\{x',y',z'\} \not\subseteq \{s'_1, s'_2\}$  with~$\prescript{}{i} \chi_{G_\top}^{k}(x,y,z) = \prescript{}{i} \chi_{G'_\top}^{k}(x',y',z')$ and~$\prescript{}{i+1} \chi_{G_\top}^{k}(x,y,z) \neq \prescript{}{i+1} \chi_{G'_\top}^{k}(x',y',z')$. 

Thus, there must be a color triple~$(c_1,c_2,c_3)$ such that the sets
\begin{align*}
M \coloneqq{} & \big\{w \mid w \in V(G_\top) \backslash \{x,y,z\},
\\
& \phantom{\big\{w \mid{}} (\prescript{}{i} \chi^k_{G_\top}(w,y,z), 
\prescript{}{i} \chi^k_{G_\top}(x,w,z), \prescript{}{i} \chi^k_{G_\top}(x,y,w)) = (c_1,c_2,c_3)\big\}
\intertext{and} 
M' \coloneqq{} & \big\{w' \mid w' \in V(G'_\top) \backslash \{x',y',z'\}, \\
& \phantom{\{w' \mid{}} (\prescript{}{i} \chi^k_{G'_\top}(w',y',z'), \prescript{}{i} \chi^k_{G'_\top}(x',w',z'), \prescript{}{i} \chi^k_{G'_\top}(x',y',w')) = (c_1,c_2,c_3)\big\}
\intertext{do not have the same cardinality. Let}
D \coloneqq{} & \{(\prescript{}{i} \chi^k_{G}(w,y,z), \prescript{}{i} \chi^k_{G}(x,w,z), \prescript{}{i} \chi^k_{G}(x,y,w)) \mid w \in M\} \, \cup{} 
\\
& \{(\prescript{}{i} \chi^k_{G'}(w',y',z'), \prescript{}{i} \chi^k_{G'}(x',w',z'), \prescript{}{i} \chi^k_{G'}(x',y',w')) \mid w' \in M'\}.
\end{align*}
By induction and by Theorem~\ref{thm:edges_conncomp} we have that
\begin{align*}
\big\{w \mid w \in V(G) \backslash \{x,y,z\}, (\prescript{}{i} \chi^k_{G}(w,y,z), \prescript{}{i} \chi^k_{G}(x,w,z),\prescript{}{i} \chi^k_{G}(x,y,w)) \in D\big\} &= M
\intertext{and}
\big\{w' \mid w' \in V(G') \backslash \{x',y',z'\}, (\prescript{}{i} \chi^k_{G'}(w',y',z'), \prescript{}{i} \chi^k_{G'}(x',w',z'),\prescript{}{i} \chi^k_{G'}(x',y',w')) \in D\big\} &= M'.
\end{align*}
Hence these sets do not have the same cardinality. Thus,~$\prescript{}{i+1} \chi_G^{k}(x,y,z) \neq \prescript{}{i+1} \chi_{G'}^{k}(x',y',z')$. 

Having shown Implication~(\ref{imp:induction2}), it suffices to show that 
\[
\{ \chi^k_{G_\top}(s_1,s_2,v) \mid v \in K\} \cap \{ \chi^k_{G'_\top}(s'_1,s'_2,v') \mid v' \in K'\} = \varnothing.
\]
For this, it suffices to prove that~$\chi^k_{G_\top}(s_1) \neq \chi^k_{G'_\top}(s'_1)$ holds. 

The graphs~$(G_{\top}^{(\{s_1,s_2\},K)}, \lambda_\top^{(\{s_1,s_2\},K)})$ and~$(G_{\top}'^{(\{s'_1,s'_2\},K')}, \lambda_\top'^{(\{s'_1,s'_2\},K')})$ are~3-connected. Thus, if we had that~$\chi^k_{G_\top}(s_1) = \chi^k_{G'_\top}(s'_1)$, there would have to be an isomorphism from the graph~$(G_{\top}^{(\{s_1, s_2\},K)}, \lambda_\top^{(\{s_1, s_2\},K)})$ to~$(G_{\top}'^{(\{s'_1, s'_2\}',K')},\lambda_\top'^{(\{s'_1, s'_2\},K')})$ that maps~$s_1$ to~$s'_1$ since we have assumed that the~$k$-dimensional WL-algorithm correctly determines orbits on all arc-colored~3-connected graphs in~$\mathcal{G}$. However, by definition of~$\lambda_\top^{(\{s_1, s_2\},K)}$ and~$\lambda_\top'^{(\{s'_1, s'_2\},K')}$, this isomorphism would also map~$s_2$ to~$s'_2$ contradicting the assumptions of the lemma.
\end{proof}

Using the lemma, we can show the following.

\begin{lemma}\label{lem:distinct_separatingpairs}
Assume~$k \geq 3$ and suppose the~$k$-dimensional Weisfeiler-Leman algorithm correctly determines orbits on all arc-colored~$3$-connected graphs in~$\mathcal{G}$. Assume~$G, G' \in \mathcal{G}$ are arc-colored~2-connected graphs of minimum degree at least~3 and let~$\{s_1,s_2\} \subseteq V(G)$ and~$\{s'_1, s'_2 \} \subseteq V(G')$ be~$2$-separators of~$G$ and~$G'$, respectively. 

If no isomorphism from~$(G_{\top}^{\{s_1, s_2\}}, \lambda_\top^{\{s_1, s_2\}})$ to~$(G_{\top}'^{\{s'_1, s'_2\}},\lambda_\top'^{\{s'_1, s'_2\}})$ maps~$s_1$ to~$s'_1$ and~$s_2$ to~$s'_2$, then~$\chi^k_{G}(s_1,s_2) \neq \chi^k_{G'}(s'_1,s'_2)$. 
\end{lemma}

\begin{proof} 
This proof is similar to the proof of Lemma~\ref{lem:distinct_cutvertices}. Suppose that
\begin{align*}
\{K_1,\dots,K_t\} & = \left\{K \mid (\{s_1,s_2\},K) \in P_0(G)\right\}
\intertext{and that}
\{K'_1,\dots,K'_{t'}\} & = \left\{K' \mid (\{s'_1,s'_2\},K') \in P_0(G')\right\}.
\end{align*}
Since there is no isomorphism from~$(G_{\top}^{\{s_1, s_2\}}, \lambda_\top^{\{s_1, s_2\}})$ to~$(G_{\top}'^{\{s'_1, s'_2\}}, \lambda_\top'^{\{s'_1, s'_2\}})$ that maps~$s_1$ to~$s'_1$ and~$s_2$ to~$s'_2$, there is an arc-colored graph~$(H,\lambda_H)$ such that the sets
\begin{align*}
I & \coloneqq \left\{j\mid \left(G_{\top}^{(\{s_1,s_2\},K_j)}, \lambda_\top^{(\{s_1,s_2\},K_j)}\right)_{(s_1,s_2)} \cong (H,\lambda_H)\right\}
\intertext{and}
I' & \coloneqq \left\{j\mid \left({G'}_{\top}^{(\{s'_1,s'_2\},K'_j)}, \lambda_\top'^{(\{s'_1,s'_2\},K'_j)}\right)_{(s'_1,s'_2)} \cong (H,\lambda_H)\right\}
\end{align*}
have different cardinalities. Note that all~$K_j$ with~$j\in I$ and all~$K'_j$ with~$j\in I'$ have the same cardinality. We know by Lemma~\ref{lem:one:component:hanging:sep:pairs} that for~$v\in K_i$ with~$i\in I$ and~$v'\in K_j$ with~$j\notin I'$ we have~$\chi^k_G(s_1,s_2,v)\neq \chi^k_{G'}(s'_1,s'_2,v')$. Letting~$C \coloneqq \{\chi^k_G(s_1,s_2,v)\mid i\in I \text{ and } v\in K_i\}$, the sets~$\{v\mid \chi^k_G(s_1,s_2,v)\in C\}$ and~$\{v'\mid \chi^k_{G'}(s'_1,s'_2,v')\in C\}$ do not have the same cardinality. We conclude that~$\chi^k_G(s_1,s_2) \neq \chi^k_{G'}(s'_1,s'_2)$.
\end{proof}

\begin{corollary}\label{cor:induced:on:bot:graphs:is:finer:3con}
	Assume~$k \geq 3$ and suppose the~$k$-dimensional Weisfeiler-Leman algorithm correctly determines orbits on all arc-colored~$3$-connected graphs in~$\mathcal{G}$. Let~$G, G' \in \mathcal{G}$ be arc-colored~2-connected graphs of minimum degree at least~3. If for vertices~$v_1,v_2 \in V(G_\perp)$ and~$v'_1,v'_2\in V(G'_\perp)$ we have~$\chi^k_{G_\perp}(v_1,v_2) \neq \chi^k_{G'_\perp}(v'_1,v'_2)$, then~$\chi^k_{G}(v_1,v_2) \neq \chi^k_{G'}(v'_1,v'_2)$.
\end{corollary}
\begin{proof}
By Lemma~\ref{lem:bot:vertices:are:different:3con}, with respect to the colorings~$\chi_G^k$ and~$\chi_{G'}^k$, the vertices in~$V(G_{\bot})$ and~$V(G'_{\bot})$ have different colors than the vertices in~$V(G) \setminus V(G_{\bot})$ and~$V(G') \setminus V(G'_{\bot})$. Thus, it suffices to show that the colorings~$\chi^k_G$ and~$\chi^k_{G'}$ refine the colorings~$\lambda_\bot$ and~$\lambda'_\bot$, respectively. For this, by the definition of~$\lambda_\bot$ and~$\lambda'_\bot$, it suffices to show the following two statements.

\begin{enumerate}
  \item If~$\{v_1,v_2\}$ and~$\{v'_1,v'_2\}$ are 2-separators and~$\ISOTYPE((G^{\{v_1,v_2\} }_{\top},\lambda^{\{v_1,v_2\}}_\top)_{(v_1,v_2)}) \neq
\ISOTYPE((G'^{\{v'_1,v'_2\} }_{\top},\lambda'^{\{v'_1,v'_2\}}_\top)_{(v'_1,v'_2)})$, then~$\chi^k_{G}(v_1,v_2) \neq \chi^k_{G'}(v'_1,v'_2)$. 
  \item If~$v_1=v_2$ and~$v'_1=v'_2$, or~$\{v_1,v_2\}\in E(G)$ and~$\{v'_1,v'_2\}\in E(G')$ but~$\{v_1,v_2\}$ and~$\{v'_1,v'_2\}$ are not 2-separators, if~$\lambda_\bot(v_1,v_2) \neq \lambda'_\bot(v'_1,v'_2)$, then~$\chi^k_{G}(v_1,v_2) \neq \chi^k_{G'}(v'_1,v'_2)$.
\end{enumerate}

The first item is exactly Lemma~\ref{lem:distinct_separatingpairs}. For the second item, from the definition of~$\lambda_\bot$ and~$\lambda'_\bot$ we obtain~$\lambda(v_1,v_2) \neq \lambda'(v'_1,v'_2)$, which implies~$\chi^k_{G}(v_1,v_2) \neq \chi^k_{G'}(v'_1,v'_2)$. 
\end{proof}

\begin{proof}[Proof of Theorem~\ref{thm:reduction3}] 
By Theorem~\ref{thm:reduction2} it suffices to show the statement for vertex-colored~2-connected graphs. To allow induction we will show the statement for arc-colored~2-connected graphs. Let~$(G, \lambda)$ and~$(G', \lambda')$ be arc-colored~2-connected graphs in~$\mathcal{G}$.
We prove the statement by induction on~$|V(G)|+|V(G')|$. If both graphs are~3-connected then the statement follows directly from the assumptions. If exactly one of the graphs is~3-connected, then exactly one of the graphs has a~2-separator and the statement follows from Corollary~\ref{cor:separating_pairs}.

Thus suppose both graphs are not~3-connected. By Lemma~\ref{lem:min:degree:3} we can assume that both graphs have minimum degree at least~3.
Since~$(G, \lambda)$ and~$(G', \lambda')$ are not isomorphic, we know by Lemma~\ref{lem:iso:if:and:only:if:smaller:iso} that~$(G_{\bot}, \lambda_\bot) \not\cong (G'_{\bot},\lambda'_\bot)$. Note that~$G_{\bot}$ and~$G'_{\bot}$ are~2-connected.
By Lemma~\ref{lem:bot:vertices:are:different:3con} the vertices in~$V(G_{\bot})$  and~$V(G'_{\bot})$ have different colors than the vertices in~$V(G) \setminus V(G_{\bot})$ and~$V(G') \setminus V(G'_{\bot})$.
Moreover by Corollary~\ref{cor:induced:on:bot:graphs:is:finer:3con}, the partition of the vertices and arcs induced by the coloring~$\chi^k_G$ restricted to~$V(G_{\bot})$  is finer than the partition induced by~$\lambda_\bot$. Similarly, the partition induced by~$\chi^k_{G'}$ on~$V(G'_{\bot})$  is finer than the partition induced by~$\lambda'_\bot$. By induction the~$k$-dimensional WL-algorithm distinguishes~$(G_{\bot}, \lambda_\bot)$ from~$(G'_{\bot},\lambda'_\bot)$. Thus the~$k$-dimensional WL-algorithm distinguishes~$(G, \lambda)$ from~$(G', \lambda')$.
\end{proof} 

In the last two sections we have concerned ourselves with graphs being distinguished (referring to two input graphs from a class) rather than graphs being identified (referring to one input graph from a class and another input graph being arbitrary). However, the theorems we prove also have 
corresponding versions concerning the latter notion.

For a graph~$G$ the \emph{Weisfeiler-Leman dimension} of~$G$ is the least integer~$k$ such that the~$k$-dimensional WL-algorithm distinguishes~$G$ from every non-isomorphic graph~$G'$.

\begin{lemma}
Let~$\mathcal{G}$  be a minor-closed graph class. The Weisfeiler-Leman dimension of graphs in~$\mathcal{G}$ is at most~$\max\{3,k\}$, where~$k$ is the minimal number~$\ell$ such that the~$\ell$-dimensional Weisfeiler-Leman algorithm correctly determines orbits on all arc-colored~$3$-connected graphs in~$\mathcal{G}$ and identifies such graphs.
\end{lemma}

The proof follows almost verbatim the lines of the entire proof of Theorem~\ref{thm:reduction3} outlined in the last two sections replacing ``distinguishes'' with ``identifies''.

\section{Arc-colored~3-connected planar graphs}\label{sec:3:con}

Let~$G$ be a~3-connected planar graph. We show that typically we can individualize two vertices in~$G$ so that applying the~$1$-dimensional WL-algorithm yields a discrete graph. There are some~3-connected planar graphs for which this is not the case. However, we can precisely determine the collection of such exceptions. 

\begin{definition}
We call a graph~$G$ an \emph{exception} if~$G$ is a~3-connected planar graph in which there are no two vertices~$v,w$ in~$G$ such that~$\chi^1_{G_{(v,w)}}$ is the discrete coloring.
\end{definition}

Here and in the following we denote by~$G_{(v_1,v_2,\dots,v_t)}$ the colored graph obtained from the (uncolored) graph~$G$ by individualizing the vertices~$v_1,v_2,\dots,v_t$ in that order. More specifically, we let~$G_{(v_1,v_2,\dots,v_t)}$ be the colored graph~$(G,\lambda)$ with
\[\lambda(v)\coloneqq 
\begin{cases} i & \text{if~$v=v_i$} \\ 0 &\text{if~$v\notin \{v_1,\dots,v_t\}$}.\end{cases}\]
As before~$\chi^1_H$ denotes the stable coloring of the~1-dimensional WL-algorithm applied to the graph~$H$.

\begin{lemma}\label{lem:3:sing:on:face:implies:discrete}
Let~$G$ be a~$3$-connected planar graph and let~$v_1,v_2,v_3$ be vertices of~$G$. If~$v_1,v_2,v_3$ lie on a common face, then~$\chi^1_{G_{(v_1,v_2,v_3)}}$ is a discrete coloring.
\end{lemma}

\begin{proof}
We will use the Spring Embedding Theorem of Tutte~\cite{MR0158387} (see \cite[Section~12.3]{DBLP:reference/crc/Kobourov13}) which is as follows. Let~$v_1,v_2,v_3$ be vertices of a common face of~$G$.
Let~$\mu_0 \colon V(G)\setminus\{v_1,v_2,v_3\} \rightarrow \mathbb{R}^2$ be an arbitrary mapping that satisfies~$\mu_0(v_1) = (0,0)$,~$\mu_0(v_2) = (1,0)$ and~$\mu_0(v_3) = (0,1)$. 
For~$i\in  \mathbb{N}$ we define~$\mu_{i+1}$ recursively by setting
\[\mu_{i+1}(v) = \begin{cases} \frac{1}{d(v)} \sum_{w\in N(v)} \mu_i(w) & \text{if~$v\notin \{v_1,v_2,v_3\}$,} \\ \mu_{i}(v) &\text{otherwise.}\end{cases} \] 
Then Tutte's result says that this recursion converges to a barycentric planar embedding of~$G$, that is, an embedding in which every vertex not in~$\{v_1,v_2,v_3\}~$ is contained in the convex hull of its neighbors \cite{MR0158387,DBLP:reference/crc/Kobourov13}. 
This implies that after a finite number of steps the barycentric embedding is injective, i.e., no two vertices are mapped to the same image.

From the theorem, we will only require the fact that for some~$i$ the map~$\mu_i$ is injective.
Choose~$\mu_0$ with the requirements above and so that all vertices in~$V(G)\setminus\{v_1,v_2,v_3\}$ have the same image. For example set~$\mu_0(v) = (1,1)$ for~$v\in V(G)\setminus\{v_1,v_2,v_3\}$.

We argue the following statement by induction on~$i$.
For every two vertices~$v$ and~$v'$, it holds that
\[\mu_i(v)\neq \mu_i(v') \quad \Rightarrow \quad \prescript{}{i}\chi^1_{G_{(v_1,v_2,v_3)}}(v)\neq \prescript{}{i}\chi^1_{G_{(v_1,v_2,v_3)}}(v'),\]
where~$\prescript{}{i}\chi^1_{G_{(v_1,v_2,v_3)}}(x)$ denotes the color of vertex~$x$ after the~$i$-th iteration when the~1-dimensional WL algorithm is applied to~$G_{(v_1,v_2,v_3)}$.  

For~$i=0$ the statement holds by the definition of~$\mu_0$ and the fact that~$v_1$,~$v_2$ and~$v_3$ are singletons in~$G_{(v_1,v_2,v_3)}$.
For~$i>0$, if~$\sum_{w\in N(v)} \mu_i(w)  \neq \sum_{w'\in N(v')} \mu_i(w')$ then the multisets~$\{\!\!\{\mu_i(w) \mid w \in N(v) \}\!\!\}$ and~$\{\!\!\{\mu_i(w') \mid w' \in N(v')\} \!\!\}$ are different and thus, by induction, the multisets~$\{\!\!\{\prescript{}{i}\chi^1_{G_{(v_1,v_2,v_3)}}(w) \mid w \in N(v) \}\!\!\}$ and~$\{\!\!\{\prescript{}{i}\chi^1_{G_{(v_1,v_2,v_3)}}(w') \mid w' \in N(v')\} \!\!\}$ are different.

We conclude with the fact that for some~$i$ the map~$\mu_i$ is injective implying that~$\prescript{}{i}\chi^1_{G_{(v_1,v_2,v_3)}}$ and therefore also~$\chi^1_{G_{(v_1,v_2,v_3)}}$ are discrete colorings.
\end{proof}

From the lemma one can directly conclude that for~$k\geq 4$, the~$k$-dimensional WL-algorithm correctly determines the orbits of every~3-connected planar graph.

\begin{corollary}\label{cor:3:con:planar:have:WL:dim:4}
For~$k\geq 4$, the~$k$-dimensional Weisfeiler-Leman algorithm  correctly determines orbits of arc-colored~3-connected planar graphs.
\end{corollary}

\begin{proof}
Let~$G$ be an arc-colored~3-connected planar graph. Then by Lemma~\ref{lem:3:sing:on:face:implies:discrete}, there are vertices~$v_1$,~$v_2$,~$v_3$ such that~$\chi^1_{G_{(v_1,v_2,v_3)}}$ is discrete (the additional arc coloring can only refine the stable coloring of the uncolored graph). This implies that the multiset~$C \coloneqq \{\!\!\{ \chi^4_{G}(v_1,v_2,v_3,x)\mid x\in V(G) \}\!\!\}$ contains~$n$ different colors. Let~$H$ be a second arc-colored~3-connected planar graph. If~$H$ contains vertices~$v_1',v_2',v_3'$ such that we have~$\{\!\!\{\chi^4_{H}(v'_1,v'_2,v'_3,x')\mid x'\in V(H) \}\!\!\} = C$, then~$G$ and~$H$ are isomorphic via an isomorphism that maps~$v_1$ to~$v'_1$. Otherwise the color~$\chi^4_{G}(v_1,v_2,v_3,v_3)$ is for all~$v'_1,v'_2,v'_3\in V(H)$ different from the color~$\chi^4_{H}(v'_1,v'_2,v'_3,v'_3)$ implying that~$G$ and~$H$ are distinguished and thus their sets of vertex colors are disjoint.
\end{proof}

This is also true for~$k=3$. The proof amounts to proving the following theorem. 

\begin{theorem}\label{WL:dim:of:3:con:graphs}
If~$G$ is an exception (i.e.,~$G$ is a~3-connected planar graph without a pair of vertices~$v,w$ such that~$\chi^1_{G_{(v,w)}}$ is discrete), then~$G$ is isomorphic to one of the graphs in Figure~\ref{fig:3:con:plan:gra:fix:3}.
\end{theorem}

Before we present the lengthy proof of the theorem, we state its implications.

\begin{corollary}\label{cor:fix:3:char}
Let~$G$ be a~3-connected planar graph. The fixing number of~$G$ is at most~3 with equality attained if and only~$G$ is isomorphic to an exception (i.e., a graph depicted in Figure~\ref{fig:3:con:plan:gra:fix:3}).
\end{corollary}

\begin{proof}
If in a given graph there is a set of~$\ell$ vertices such that individualizing all vertices in the set and then applying the~1-dimensional WL-algorithm yields a discrete coloring, then Theorem~\ref{WL:dim:of:3:con:graphs}, a graph that is not an exception has fixing number at most~2. To conclude the corollary it thus suffices to check that all exceptions have fixing number~3.  
\end{proof}

\begin{corollary}\label{cor:3:con:planar:have:WL:dim:3}
For~$k\geq 3$, the~$k$-dimensional Weisfeiler-Leman algorithm correctly determines orbits of arc-colored~3-connected planar graphs.
\end{corollary}

\begin{proof}
Suppose~$G$ is an arc-colored~3-connected planar graph that is not an exception. Then there are vertices~$v$ and~$w$ such that~$\chi^1_{G_{(v,w)}}$ is discrete. Analogously to the proof of Corollary~\ref{cor:3:con:planar:have:WL:dim:4}, we obtain that for any second arc-colored~3-connected planar graph~$H$, the~$3$-dimensional WL-algorithm only assigns equal colors to a vertex of~$G$ and a vertex of~$H$ if there is an isomorphism mapping the one to the other. 

With direct computations, one can check that each arc-colored exception is distinguished from all arc-colored~3-connected planar graphs and that on each arc-colored exception the stable coloring under the~$3$-dimensional Weisfeiler-Leman algorithm induces the orbit partition on the vertices.
\end{proof}

The task in the rest of this section is to show Theorem~\ref{WL:dim:of:3:con:graphs}. The proof of the theorem is a lot more involved than the proof of Corollary~\ref{cor:3:con:planar:have:WL:dim:3}. Thus, at the expense of increasing~$k$ by~1 from~3 to~4 in the main theorem (Theorem~\ref{thm:main}), the reader may skip the following lengthy exposition. 

We determine the exceptions in a case-by-case analysis with respect to the existence of vertices of certain degrees. 
The following two lemmas serve as general tools to deduce information about the structure of these input graphs. 

For a subgraph~$G'$ of a graph~$G$ we say that~$v\in G'$ is \emph{saturated in~$G'$ with respect to~$G$} if~$d_{G'}(v)= d_{G}(v)$. Thus, if a vertex is saturated, then its neighbors in~$G$ and in~$G'$ are the same.

\begin{lemma}\label{lem:saturation:implies:face:cycle}
Let~$G$ be a~3-connected planar graph.
\begin{enumerate}
\item\label{item:1:of:saturation:lemma} Let~$G'$ be a subgraph of~$G$ and suppose that the sequence~$v_1,\dots, v_t$ forms a face cycle in the planar embedding of~$G'$ induced by a planar embedding of~$G$. If in~$\{v_1,\dots,v_t\}$ there are at most two vertices that are not saturated in~$G'$ with respect to~$G$, then~$V(G) = \{v_1,\dots,v_t\}$ or~$v_1,\dots,v_t$ is a face cycle of~$G$.
\item\label{item:2:of:saturation:lemma}  If~$v_1,\dots,v_t$ is a~3-cycle or an induced~4-cycle of~$G$ that contains two vertices of degree~$3$ in~$G$, then~$V(G) = \{v_1,\dots,v_t\}$ or~$v_1,\dots,v_t$ is a face cycle of~$G$.
\end{enumerate}
\end{lemma}

\begin{proof}[Proof sketch]
For Part~\ref{item:1:of:saturation:lemma}, suppose~$v_1,\dots, v_t$ forms a face cycle of~$G'$ and there are at most two vertices~$v_i$ and~$v_j$ that are not saturated in~$G'$ with respect to~$G$. 
Assume~$v_1,\dots, v_t$ is not a face cycle in~$G$. Then the vertices~$v_i$ and~$v_j$ are the only ones among~$v_1,\dots,v_t$ that have neighbors in~$G$ inside the region of the plane corresponding to the face cycle of~$G'$ formed by~$v_1,\dots,v_t$. Therefore, if ~$V(G)\neq \{v_1,...,v_t\}$, then~$\{v_i,v_j\}$ is a separator of~$G$, which contradicts the~3-connectivity.

For Part~\ref{item:2:of:saturation:lemma} consider an induced~4-cycle~$v_1,\dots,v_4$ which contains two vertices~$v_i$ and~$v_j$ that have degree~3 in~$G$. If~$v_1,\dots,v_4$ is not a face cycle of~$G$, then~$\{v_1,\dots,v_4\}\setminus \{v_i,v_j\}$ is a separator of size~2.
The argument for a~3-cycle~$v_1,\dots,v_3$ is similar.
\end{proof}

\begin{lemma}\label{lem:consecutive:neighbors:have:common:neighbor}
Let~$G$ be an exception and let~$v$ be a vertex of~$G$. Let~$u_1, \dots, u_{d(v)}$ be the cyclic ordering of the neighbors of~$v$ induced by a planar embedding of~$G$. Then every pair of vertices~$u_i, u_{i+1}$ has a common neighbor of degree~$d(v)$ other than~$v$.
\end{lemma}

\begin{proof}
If~$u_i$ and~$u_{i+1}$ do not have a common neighbor of degree~$d(v)$ other than~$v$, the coloring~$\chi^1_{G_{(u_i,u_{i+1})}}$ has the three singletons~$u_i, v, u_{i+1}$, which lie on a common face. Thus, by Lemma~\ref{lem:3:sing:on:face:implies:discrete}, the coloring~$\chi^1_{G_{(u_i,u_{i+1})}}$ is discrete, contradicting the assumption that~$G$ is an exception.
\end{proof}

Now we can start determining the structure of the exceptions.

\begin{lemma}\label{tri:con:2:individ:degree:5}
If~$G$ is an exception that has a vertex of degree~5, then it is isomorphic to the icosahedron or the bipyramid on~7 vertices.
\end{lemma}

\begin{proof}
To simplify notation, we let~$\chi_G \coloneqq \chi^1_G$ in the course of this proof. 
Let~$G$ be an exception with a vertex~$v$ of degree~5. Let~$N \coloneqq N(v)$ be the set of neighbors of~$v$ and let~$(u_1,\dots,u_5)$ be their circular ordering. For convenience we will take indices modulo~5. 

By Lemma~\ref{lem:consecutive:neighbors:have:common:neighbor}, every pair~$u_i,u_{i+1}$ has a common neighbor~$x_{i,i+1}$ of degree~$5$ other than~$v$. (We remark that the vertices~$x_{i,i+1}$ are not necessarily distinct or unique.)

\begin{claim}\label{claim:1:of:tricon:lemma}{For all~$i\in \{1,\dots,5\}$, every pair of vertices~$u_i,u_{i+2}$ has a common neighbor~$x_{i,i+2}$ of degree~$5$ other than~$v$.} 

\proof
To show the claim assume without loss of generality that~$u_1$ and~$u_3$ do not have a common neighbor of degree~$5$ other than~$v$. Consider the coloring~$\chi_{G_{(u_1,u_{3})}}$.
In this coloring the vertex~$v$ is a singleton.  It follows that~$N$ is the union of three color classes of the coloring~$\chi_{G_{(u_1,u_3)}}$, one of which is~$\{u_2, u_4, u_5\}$. (Otherwise, there are two consecutive vertices in~$N$ that are singletons and thus~$\chi_{G_{(u_1,u_{3})}}$ is discrete by Lemma~\ref{lem:3:sing:on:face:implies:discrete}.)

\renewcommand{\outerCase}{\Alph}
\begin{proofcases}
\proofcase[$x_{1,2} \in N$ or\,~$x_{2,3} \in N$] \label{case:A:of:claim:1:of:tricon:lemma}
We only consider the case that~$x_{1,2} \in N$, since the case~$x_{2,3} \in N$ is analogous.  We know that~$u_2$,~$u_4$ and~$u_5$ have the same degree in~$G[N]$ and thus the vertices~$u_1$ and~$u_3$ must have the same degree in~$G[N]$ since otherwise one of them would have a unique degree in~$G[N]$ (and then some~$u_i$ would be a singleton in~$\chi_{G_{(v)}}$ and thus~$\chi_{G_{(v,u_{i+1})}}$ would be discrete by Lemma~\ref{lem:3:sing:on:face:implies:discrete}).

Now suppose first~$x_{1,2} = u_4$ or~$x_{1,2} = u_5$. Either way, one vertex and thus all the vertices~$u_2,u_4,u_5$ are adjacent to~$u_1$, since they form a color class of~$\chi_{G_{(u_1,u_{3})}}$. It follows that the degree of~$u_1$ and hence of~$u_3$ in~$G[N]$ is at least~3. Thus,~$u_1$ and~$u_3$ have a common neighbor. It has degree~5 since~$x_{1,2}$ has degree~5. The case~$x_{1,2} = u_3$ can be treated analogously by simply swapping the roles of~$u_1$ and~$u_3$. 

\proofcase[$x_{1,2} \notin N$ and~$x_{2,3} \notin N$]\label{case:B:of:claim:1:of:tricon:lemma}
Since~$u_2$ and~$u_4$ have the same color, the vertices~$u_1$ and~$u_4$ have a common neighbor~$x\notin N$ of degree~$5$ other than~$v$. Similarly~$u_3$ and~$u_5$ have a common neighbor~$x'\notin N$ of degree~$5$ other than~$v$. By the planarity of~$G$, we have~$x= x'$, and thus~$x$ is a vertex of degree~5 adjacent to~$u_1$ and to~$u_3$.
\uend
\end{proofcases}
\renewcommand{\outerCase}{\arabic}
\end{claim}

Having proved the claim we now finish the proof of Lemma~\ref{tri:con:2:individ:degree:5}.
Again we distinguish cases.

\begin{proofcases}
\proofcase[${G[N]}$ is non-empty and some vertex in~$N$ has degree~5 in~$G$]\label{case:1:of:tricon:lemma}

Due to plana\-rity there can be at most one vertex in~$N$ that has degree~4 within~$G[N]$. Indeed, if there were two such vertices~$u_i$ and~$u_j$, then each of~$v$,~$u_i$,~$u_j$ would be adjacent to all vertices in~$N \backslash \{u_i, u_j\}$, yielding a~$K_{3,3}$ minor. However, no vertex in~$N$ can have a unique degree in~$G[N]$ and thus~$G[N]$ has a maximum degree of at most~3. Suppose that~$G[N]$ contains an edge~$\{u_i,u_{i+2}\}$ for some~$i$, say~$i=1$. Due to Claim~\ref{claim:1:of:tricon:lemma}, the vertex~$u_{2}$ must share a common neighbor with each of~$u_4$ and~$u_5$ apart from~$v$. By planarity these common neighbors are in~$\{u_1,u_3\}$. However the two common neighbors must be different since otherwise the respective vertex has degree~$4$ in~$G[N]$. Thus~$u_2$ is adjacent to~$u_1$ and~$u_3$ and both~$u_1$ and~$u_3$ have degree~3 in~$G[N]$. 

By the planarity of~$G$, the vertex~$u_2$ has degree~2 in~$G[N]$. Suppose~$u_4$ has degree~$3$ in~$G[N]$. Then it must be adjacent to~$u_1$,~$u_3$ and~$u_4$. Since~$u_5$ cannot have a unique degree in~$G[N]$, it must be adjacent to~$u_3$, i.e., it must have degree~$2$ in~$G[N]$. (By planarity, it cannot be adjacent to~$u_1$.) However, this would give~$u_3$ a degree of~4 in~$G[N]$, yielding a contradiction. The case that~$u_5$ has degree~$3$ in~$G[N]$ is symmetric. Thus,~$u_1$ and~$u_3$ are the only vertices of degree~3 in~$G[N]$.

Then~$u_2$ is the only vertex of~$N$ that is adjacent to two vertices of degree~3 in~$G[N]$, making~$u_2$ a singleton in~$\chi_{G_{(v)}}$ and yielding a contradiction.

We conclude that there is no edge of the form~$\{u_i,u_{i+2}\}$. This implies that in~$G[N]$ there is no vertex of degree~1. (Otherwise, we could individualize this vertex and~$v$, and together with their unique common neighbor, they would yield a discrete coloring by Lemma~\ref{lem:3:sing:on:face:implies:discrete}.)  Consequently, since~$G[N]$ is non-empty, we conclude that~$u_1,u_2,\dots,u_5$ is an induced cycle in~$G[N]$. Since some vertex in~$N$ has degree~$5$ and within~$G[N]$ the two neighbors of each vertex must have the same degree, we conclude that all vertices in~$N$ have degree~$5$.

Thus, if a vertex fulfills Case~\ref{case:1:of:tricon:lemma}, then all its neighbors also fulfill Case~\ref{case:1:of:tricon:lemma}. Therefore, being connected, the entire graph~$G$ must be a~5-regular triangulated planar graph since we have restricted ourselves to connected graphs.
Being~5-regular, the graph has~$m= 5/2 n$ edges, and being a triangulation the number of edges is~$m=3n-6$. We conclude that~$n=12$.
There is only one~5-regular graph on~12 vertices, the icosahedron (see for example \cite{DBLP:journals/jgaa/HasheminezhadMR11}).
 
\proofcase[${G[N]}$ is empty or no vertex in~$N$ has degree~$5$ in~$G$]\label{case:2:of:tricon:lemma} 

From Claim~$\ref{claim:1:of:tricon:lemma}$ we already know that every pair of vertices~$u_i, u_{i+2}$ has a common neighbor of degree~5 other than~$v$. In Case~\ref{case:2:of:tricon:lemma} this vertex cannot be in~$N$. Due to planarity, all these common neighbors for different~$i$ must be equal to a single vertex~$y$ adjacent to all vertices of~$N$.

Observing that~$y$ has degree~$5$, consider the subgraph~$H \coloneqq G[\{v,y,u_1,\dots,u_5\}]$ of~$G$. With the described circular ordering of the vertices in~$N$ there is only one planar drawing of~$H$ up to equivalence. In this drawing every face is a~4-cycle or a~$3$-cycle containing~$y$ and~$v$. Since~$y$ and~$v$ have degree~$5$ in~$G$, but they already have degree~5 in~$H$, they are saturated. 
Thus, since both~$y$ and~$v$ belong to each face of~$H$, by Lemma~\ref{lem:saturation:implies:face:cycle}, no interior of a face of the drawing of~$H$ contains vertices of~$G$. Therefore~$G = H =  G[\{v,y,u_1,\dots,u_5\}]$. Since~$G$ is~3-connected,~$G[N]$ cannot be empty. Similarly as in Case~\ref{case:1:of:tricon:lemma} we conclude that~$u_1,u_2,\dots,u_5$ is a cycle rendering~$G$ the bipyramid on~$7$ vertices.
\QED
\end{proofcases}
\end{proof}

\begin{lemma}\label{tri:con:2:individ:degree:3}
If~$G$ is an exception that has a vertex of degree~3, then it is isomorphic to a tetrahedron, a cube, a triangular bipyramid, a triakis tetrahedron, a rhombic dodecahedron, or a triakis octahedron.
\end{lemma}

\begin{proof}
Assume~$G$ is a~3-connected planar graph with a vertex of degree~3 and that~$G$ does not have two vertices, individualization of which followed by an application of the~1-dimensional WL-algorithm produces the discrete partition.
 
Let~$v$ be a vertex of degree~3 in~$G$ and
let~$N\coloneqq N(v)= \{u_1,u_2,u_3\}$ be its neighbors. 
By Lemma~\ref{lem:3:sing:on:face:implies:discrete}, no vertex of~$N$ can have a unique degree in~$G$. Thus, the graph~$G[N]$ is either a triangle or empty.

By Lemma~\ref{lem:consecutive:neighbors:have:common:neighbor}, for~$i\in \{1,2,3\}$ (indices always taken modulo~3) the pair~$u_i,u_{i+1}$  has a common neighbor~$x_{i,i+1}$ of degree~$3$ other than~$v$. (As in the previous proof, these~$x_{i,i+1}$ are not necessarily distinct or unique.) If~$x_{i,i+1}\in N$ then~$x_{i,i+1}= u_{i+2}$ and~$N(\{u_{i+2},v\}) = \{u_i,u_{i+1}\}$. Thus, unless~$G$ only has~4 vertices (in which case it is the tetrahedron), the set~$\{u_i,u_{i+1}\}$ forms a~2-separator, which contradicts~$G$ being~3-connected. We can therefore assume for all~$i\in \{1,2,3\}$ that~$x_{i,i+1} \notin N$.

We claim that~$v, u_i,x_{i,i+1}, u_{i+1}$ is a face cycle of~$G$, or~$u_i$ and~$u_{i+1}$ are adjacent and both~$u_i,u_{i+1},x_{i,i+1}$ and~$u_i,u_{i+1},v$ are face cycles. The vertex~$x_{i,i+1}$ has degree~3, so the claim follows directly from Part~\ref{item:2:of:saturation:lemma} of Lemma~\ref{lem:saturation:implies:face:cycle}.

\begin{proofcases}
\proofcase[${G[N]}$ is empty]\label{case:1:of:tricon23:lemma}

In this case every face incident to~$v$ is a~4-cycle consisting of two non-adjacent vertices of degree~3 in~$G$ and two other vertices of degree~$d\geq 3$. By analogous arguments as for~$v$, for every~$i \in \{1,2,3\}$, the graph~$G[N(x_{i,i+1}]$ must be either empty or a triangle. It cannot be a triangle because the edge~$\{u_i,u_{i+1}\}$ is not present. So by successively replacing~$v$ by its opposite vertices in the incident~4-cycles, we conclude that~$G$ is a~$(3,d)$-bi-regular quadrangulation.  

\begin{proofcases}
\proofcase[$G$ is~3-regular]\label{case:1.1:of:case:1:of:tricon23:lemma} 
In this subcase~$G$ must be a~3-regular quadrangulation. Such a graph has~$m= 3n/2$ edges, since it is~3-regular, but also~$m= 2n-4$ edges since it is a quadrangulation.
Thus~$n = 8$. It is easy to verify that the only~3-regular planar quadrangulation on~8 vertices is the cube.

\proofcase[$G$ is not~3-regular]\label{case:1.2:of:case:1:of:tricon23:lemma} 
Then~$G$ is bipartite and bi-regular with degrees~$3$ and~$d$ say. Let~$n_3$ and~$n_d$ be the number of vertices of degree~$3$ and~$d$, respectively. Then~$3 n_3 = d n_d$ by double counting and~$d n_d = m = 2n-4$ since~$G$ is a quadrangulation. It follows that~$d n_d = 2(n_3+n_d)-4 = 2(d n_d/3+n_d)-4$ which gives that~$4 = n_d (2-d/3)$. Thus~$d\leq 5$. The case~$d=3$ is Case~\ref{case:1.1:of:case:1:of:tricon23:lemma} and~$d=5$ cannot occur according to Lemma~\ref{tri:con:2:individ:degree:5}.

We conclude that~$G$ is a~$(3,4)$-biregular quadrangulation. We have~$3n_3 = 4n_4$. Then~$m = 3 n_3 = 3 \cdot 4/7 n$. But also~$m = 2n-4$. Thus~$n= 14$. If we modify~$G$ by adding an edge between every pair of degree~$3$ vertices that share a common face and by removing all vertices that originally had degree~$4$, we obtain a new graph~$G'$ that is a~$3$-regular planar quadrangulation on~$8$ vertices.  The only such graph is the cube. Undoing the modification we obtain that~$G$ is the rhombic dodecahedron.
\end{proofcases}

\proofcase[${G[N]}$ is a triangle]\label{case:2:of:tricon23:lemma}

In this case every face in~$G$ is a~3-cycle consisting of a vertex of degree~3 and two other vertices of equal degree~$d$.

\begin{proofcases}
\proofcase[$G$ is~3-regular, i.e.,~$d=3$]\label{case:2.1:of:case:2:of:tricon23:lemma}
Then~$G$ is a~3-regular triangulation. Thus~$m = 3/2 n$ and~$m = 3n-6$, thus~$n= 4$. We conclude that~$G$ is the tetrahedron.

\proofcase[$G$ is not~3-regular]\label{case:2.2:of:case:2:of:tricon23:lemma}
Consider the graph~$G'$ obtained from~$G$ by removing all vertices of degree~$3$. The resulting graph is a planar~$d/2$-regular triangulation. Indeed, since~$v$ was chosen arbitrarily among all vertices of degree~$3$ and the vertices~$x_{i,i+1}$ have degree~$3$ as well, it is easy to see that the resulting graph is a triangulation. Moreover, one can verify that in~$G$ for every vertex of degree~$d$, in the cyclic ordering of its neighbors, the degrees~$3$ and~$d$ alternate. Thus, a deletion of the vertices of degree~3 halves the degrees of each of the other vertices.

It follows that~$d/2\in \{2,3,4,5\}$. For~$d/2 =2$ we obtain a~3-cycle. This implies that~$G$ is a triangular bipyramid. For~$d/2=3$ we conclude that~$G'$ is a tetrahedron. This implies that~$G$ is a triakis tetrahedron. For~$d/2=4$ we obtain an octahedron. This implies that~$G$ is the triakis octahedron. For~$d/2=5$ we would obtain an icosahedron. This would imply that~$G$ is the triakis icosahedron. However in this solid, there are two vertices~$u,u'$ (namely vertices of degree~3 of distance~2 that only have one common neighbor) such that~$\chi^1_{G_{(u,u')}}$ is discrete.
\QED
\end{proofcases}
\end{proofcases}
\end{proof}

\begin{lemma}\label{tri:con:2:individ:degree:4}
If~$G$ is an exception that has a vertex of degree~4, then it is isomorphic to a bipyramid, a rhombic dodecahedron or a tetrakis hexahedron.
\end{lemma}

\begin{proof}
Assume~$G$ is an exception with a vertex of degree~4. First we make two observations that hold for every vertex~$u$ of degree~4 with neighbors~$v_1,v_2,v_3,v_4$ in cyclic order.

\emph{Observation~1:}\phantomsection\label{obs:1} \ It holds that~$d(v_1)= d(v_3)$ and~$d(v_2)= d(v_4)$. Otherwise we can individualize~$u$ and a neighbor~$v_i$ of~$u$ so that~$v_{i+1}$ or~$v_{i-1}$ refines to a singleton class, which yields a contradiction. 

\emph{Observation~2:}\phantomsection\label{obs:2} \ By a similar argument, the induced graph~$G[v_1,v_2,v_3,v_4]$ either is empty or forms an induced cycle such that~$v_i$ is adjacent to~$v_{i+1}$  for all~$i\in \{1,2,3,4\}$ (indices  taken modulo~4). 

Due to Observation~\hyperref[obs:2]{2}, if~$G$ is~$4$-regular then~$G$ is either a triangulation or every face is of size at least~4. In the first case~$G$ has~$n=6$ vertices (since~$m= 4n/2$ and~$m=3n-6$), and thus~$G$ is the octahedron (a bipyramid).

The second case cannot occur since a planar graph without triangle faces has at most~$2n-4$ edges, but a~4-regular graph has~$2n$ edges. 

We can thus assume that~$G$ has a vertex~$v$ of degree other than~$4$ that is adjacent to a vertex of degree~$4$.
By Lemmas~\ref{tri:con:2:individ:degree:5} and~\ref{tri:con:2:individ:degree:3} we can assume that~$G$ neither has a vertex of degree~$5$ nor a vertex of degree~$3$.
Thus, we can assume that~$v$ has degree at least~$6$. Let~$N\coloneqq  N(v)$ be the set of neighbors of~$v$ and let~$N_4\subseteq N$ be those neighbors of~$v$ that are of degree~4. Suppose~$(u_1,\dots,u_t)$ is the cyclic order of~$N_4$ induced by the cyclic order of~$N$.

\begin{proofcases}
\proofcase[${G[N_4]}$ is non-empty]\label{case:1:of:tricon24:lemma} 
First assume there are distinct~$u,u'\in  N_4$ that are adjacent, i.e.,~$G[N_4]$ is non-empty. According to Observation~\hyperref[obs:1]{1}, each vertex in~$N_4$ must have two neighbors in~$G$ of degree~$d(v)\neq 4$, and thus~$G[N_4]$ has maximum degree~2.
We argue that~$G[N_4]$ cannot have a vertex of degree~$1$. Indeed, if~$u_i$ were a vertex that has degree~1 in~$G[N_4]$ then~$u_i$ would have two neighbors of degree~$d(v)$ and two neighbors of degree~$4$, one of which adjacent to~$v$ and one of which non-adjacent to~$v$. This means that in~$\chi^1_{G_{(v,u_i)}}$ all neighbors of~$u_i$ are singletons, since the two neighbors of~$u_i$ of degree~4 disagree on being adjacent to~$v$ or not. This is impossible by Lemma~\ref{lem:3:sing:on:face:implies:discrete}. We conclude that~$G[N_4]$ has only vertices of degree~2 and~0. Since~$G[N_4]$ is non-empty, this implies that there is some cycle in~$G[N_4]$. Assume this cycle has an edge~$\{u_i,u_j\}$ connecting two vertices that are not consecutive in the cyclic order~$(u_1,\dots,u_t)$. 

Let~$u_i^+$ and~$u_i^-$ be the vertices of~$G[N]$ following and preceeding, respectively, the vertex~$u_i$ in the cyclic ordering of~$N$ (so they may or may not have degree~4). By Lemma~\ref{lem:consecutive:neighbors:have:common:neighbor}, there must be vertices~$x^+$ and~$x^-$ of degree~$d(v)\neq 4$ such that~$x^+$ is adjacent to both~$u_i$ and~$u_i^+$ and~$x^-$ is adjacent to both~$u_i$ and~$u_i^-$. However,~$x^+\neq x^-$ since the cycle~$v,u_i,u_j$ separates~$u_i^+$ from~$u_i^-$. We conclude that~$u_i$ has the following five neighbors: the vertex~$v$, two neighbors in~$N_4$ as well as~$x^+$ and~$x^-$. But~$u_i$ has degree~4, which gives a contradiction.
We conclude that~$u_i$ is adjacent to~$u_{i+1}$ for all~$i\in \{1,\dots,t\}$.

Finally, by Observation~\hyperref[obs:1]{1}, every pair of vertices~$\{u_i, u_{i+1}\}$ must have a common neighbor other than~$v$ of degree~$d(v)$ (otherwise all neighbors of~$u_i$ would be singletons in~$\chi^1_{G_{(u_i,u_{i+1})}}$).

Since~$u_i$ is adjacent to~$v,u_{i+1}$ and~$u_{i-1}$, it can only have one further neighbor. Thus, all these common neighbors for the pairs~$\{u_i, u_{i+1}\}$ for~$i\in \{1,\dots ,t\}$ are indeed the same vertex~$x$. Consider~$G' \coloneqq G[\{u_1,\dots,u_t,v,x\}]$. The graph~$G'$ is~3-connected and every face is a~$3$-cycle with two vertices saturated in~$G'$ with respect to~$G$. Part~\ref{item:1:of:saturation:lemma} of Lemma~\ref{lem:saturation:implies:face:cycle} implies that~$G=G'$. We conclude that~$G$ is isomorphic to a bipyramid.

\proofcase[${G[N_4]}$ is empty]\label{case:2:of:tricon24:lemma}
In the second case we now assume that the degree~4 neighbors of~$v$ form an independent set.

\begin{claim}{For every~$i\in \{1,\dots,t\}$ there is a vertex~$x_{i,i+1}$ such that either the sequence~$v,u_i,x_{i,i+1},u_{i+1}$ forms a face cycle or both~$v,u_i,x_{i,i+1}$  and~$v,u_{i+1},x_{i,i+1}$ are face cycles.}\label{claim:faces:or:diagonals}

\proof
Without loss of generality, we show the claim for~$i=2$.
We first argue that there are vertices~$u'\in N_4$ and~$x'$ such that~$v,u_2,x',u'$ is a~4-cycle. For this
let~$x'$ be the first neighbor following~$v$ in the cyclic ordering among the neighbors of~$u_2$. Then by Lemma \ref{lem:consecutive:neighbors:have:common:neighbor}, there must be a vertex~$u'$ other than~$u_2$ of degree~$4$ that is adjacent to~$x'$ and~$v$. Let~$u_j\in N_4$ be the first neighbor of~$v$ following~$u_2$ in the cyclic ordering of vertices in~$N_4$ that has a common neighbor with~$u_2$ other than~$v$.

 We choose a common neighbor~$x$ of~$u_2$ and~$u_j$ so that it is closest to~$v$: more precisely, for a common neighbor~$x\neq v$ of~$u_2$ and~$u_j$ consider the cycle~$u_2,v,u_j,x$. It bounds two areas, one of which contains the vertices of~$N$ that follow~$u_2$ but precede~$u_{j}$ while the other one contains the vertices that follow~$u_{j}$ but precede~$u_{2}$ in the cyclic order of~$N$. (One of these sets may be empty.) We choose~$x$ so that the first of these areas is minimal with respect to inclusion and we call this area~$A$.

We claim that the~$4$-cycle~$u_j, x,u_2, v$ is a face cycle of~$G$ or a face cycle after removing the diagonal~$\{x,v\}$ (i.e.,  the~3-cycles~$v,u_2,x$ and~$v,u_j,x$ are faces). (Note that the edge~$\{u_2,u_j\}$ cannot be present since~$G[N_4]$ is empty.) Indeed, suppose that~$u_2$ has a neighbor that lies within~$A$. Choose as such a neighbor~$z$ the vertex that precedes~$v$ in the cyclic ordering of neighbors of~$u_2$. Then for some vertex~$\bar u\in N_4$, the sequence~$(v,u_2,z, \bar u)$ forms a~$4$-cycle (Lemma~\ref{lem:3:sing:on:face:implies:discrete} applied to~$\chi^1_{G_{(v,z)}}$). Now, either~$\bar u$ precedes~$u_j$ in the cyclic ordering of~$N_4$ starting from~$u_2$ or~$\bar u= u_j$, but~$\{v,u_2,z, u_j\}$ bounds an area that is a proper subset of~$A$. Either case contradicts the minimal choices of~$u_j$ and~$x$.
 
Finally assume that~$u_j$ has a neighbor~$z$ that lies within~$A$. Choose~$z$ to be the vertex that follows~$v$ in the cyclic ordering of neighbors of~$u_j$. Since~$\{v,x\}$ is not a separator and~$u_2$ does not have a neighbor inside~$A$, there must be a path from~$z$ to~$u_2$ via~$u_j$ that leaves~$A$. Thus, the vertex~$u_j$ must have some neighbor outside of~$A$. Hence, since~$u_j$ has degree 4, the vertex~$z$ is the only neighbor of~$u_j$ inside~$A$.

By Lemma~\ref{lem:consecutive:neighbors:have:common:neighbor}, for some vertex~$\bar u\in N_4$, the sequence~$(v,u_j,z, \bar u)$ must be a~4-cycle. Consider the coloring~$\chi^1_{G_{(u_2,\bar u)}}$. In this coloring~$v$ is a singleton, since by the minimality of~$u_j$ it is the only common neighbor of~$u_2$ and~$\bar u$. Furthermore,~$u_j$ is the only vertex in~$N_4$ that has simultaneously a common neighbor with~$u_2$ other than~$v$ and a common neighbor with~$\bar u$ other than~$v$ and thus~$u_j$ is a singleton. Moreover,~$\bar u$ and~$u_j$ only have one common neighbor other than~$v$, namely~$z$, which is then a singleton as well. The singletons~$v,u_j$ and~$z$ lie on a common face by the choice of~$z$ which yields a contradiction with Lemma~\ref{lem:3:sing:on:face:implies:discrete}. 

Thus, neither~$u_2$ nor~$u_j$ have a neighbor inside the cycle. This implies that the cycle~$u_j, x,u_2, v$ is a face or it becomes a face after removing the possible diagonal~$\{x,v\}$, since otherwise the set~$\{x,v\}$ would be a separator. We conclude that~$u_j=u_3$ and that the vertex~$x_{2,3}\coloneqq x$ justifies the claim. 
\uend
\end{claim}

In the following we call an edge of~$G$ a \emph{diagonal} if neither of its endpoints has degree~4.

Overall the claim implies that at least every second neighbor of~$v$ is of degree~4 (in particular~$|N|\leq 2|N_4|$) and thus, being of degree at least~6, the vertex~$v$ has at least~3 neighbors of degree~4, i.e.,~$|N_4|\geq 3$.

Recall that~$u_1$ and~$u_3$ are the two vertices in~$N_4$ that are closest to~$u_2$ in the cyclic ordering of neighbors of degree~4 of~$v$.

We distinguish several cases according to the size of~$N_4$.

\begin{proofcases}
\proofcase[$|N_4| = 3$]\label{case:2.1:of:case:2:of:tricon24:lemma}  

Since~$v$ must have degree at least~6 and at least every second of its neighbor has degree~4, we conclude that~$v$ also has exactly three neighbors of degree larger than~$4$.
Thus, the degree of~$v$ is~6. Let~$(u_1,t_1,u_2,t_2,u_3,t_3)$ be the neighbors of~$v$ in cyclic order. Then by Claim~\ref{claim:faces:or:diagonals}, the vertices~$u_i,t_i,v$ form a face cycle for every~$i\in \{1,2,3\}$. Likewise~$t_i,u_{i+1},v$ is a face cycle. Thus, the graph induced by~$N \cup \{v\}$ is a wheel with~$7$ vertices. By Observation~\hyperref[obs:2]{2} at the start of the proof, the neighborhood of~$u_i$ forms a cycle. Moreover, by Observation~\hyperref[obs:1]{1}, the vertices~$t_1$,~$t_2$ and~$t_3$ all have the same degree~$d$. We argue that~$d=6$. Since all~$u_i$ have degree~4, if every pair of vertices~$u_i$,~$u_{i+1}$ had a common neighbor other than~$v$ and~$t_i$, this would have to be a single vertex adjacent to~$u_1$,~$u_2$ and~$u_3$. However, such a vertex~$x$ does not exist because otherwise~$G'\coloneqq G[u_1,u_2,u_3,t_1,t_2,t_3,x,v]$ would be a~3-connected graph in which every face is a triangle with a saturated vertex or a~4-cycle with two saturated vertices and Lemma~\ref{lem:saturation:implies:face:cycle} would imply~$G = G'$ which cannot be since~$G'$ has a vertex of degree~3, namely~$x$. 

Thus, some pair~$\{u_i,u_{i+1}\}$ does not have a common neighbor other than~$v$ or~$t_i$, which in turn implies~$d(v) = d(t_i)= 6$ using Lemma~\ref{lem:consecutive:neighbors:have:common:neighbor}. Therefore, all neighbors of~$v$ have degree~$4$ or~$6$. More strongly, we conclude that the vertex degrees that appear among the neighbors of~$t_i$ are the same as the vertex degrees appearing as neighbors of~$v$, including multiplicites. Thus, each~$t_i$ also has~$3$ neighbors of degree~4. 

By Observation~\hyperref[obs:1]{1} all vertices in~$N_4$ only have neighbors of degree~$6$ (since we already know that they have~3 neighbors of degree~6). By the same argument, these degree~6 vertices have themselves~$3$ neighbors of degree~$4$ and~$3$ neighbors of degree~$6$.

We conclude that the entire graph has only vertices of degree~$4$ and~$6$. Since every face incident with~$v$ is a triangle, and~$v$ is arbitrary among the degree~$6$ vertices, we conclude that~$G$ is a triangulation. Moreover, every vertex of degree~$4$ has exactly~4 neighbors of degree~$6$ and every vertex of degree~$6$ has exactly~3~neighbors of degree~$4$. We conclude that~$4n_4= 3n_6$ where~$n_i$ is the number of vertices of~$G$ of degree~$i$. Since~$n_4+n_6 = n =|G|$ we conclude that~$G$ has~$18n/7$ edges. Since~$G$ is a triangulation, it has~$3n-6$ edges. We conclude that~$18n/7 = 3n-6$ and thus~$G$ is a graph on~$14$ vertices. Furthermore, the graph~$G'$ induced by the vertices of degree~$6$ is a~3-regular graph on~8 vertices. All faces in the induced drawing of~$G'$ are~4-cycles.
We conclude that~$G'$ is the cube. (There is only one triangle-free planar~3-regular graph on~8 vertices.) Each face of~$G'$ contains, within~$G$, a vertex of degree~$4$. We conclude that~$G$ is the tetrakis hexahedron.

\proofcase[$|N_4| = 4$]\label{case:2.2:of:case:2:of:tricon24:lemma}  

In this case~$v$ must be incident to some diagonals, since otherwise~$v$ has degree~4. Thus, for every~$i \in \{1,2,3,4\}$ the diagonal~$\{v, x_{i,i+1}\}$ must be present, since by Observation~\hyperref[obs:2]{2} the neighbors of a degree~$4$ vertex form a cycle or an independent set. It follows that~$v$ must have degree~$8$. We obtain a graph that has at most as many vertices of degree~4 as it has vertices of degree at least~8. (Every vertex of degree at least~8 is adjacent to at least~4 vertices of degree~4 and every vertex of degree~4 is adjacent to~4 vertices of degree at least~8.) Double counting implies that the graph has at least~$3n$ edges, which is impossible for a planar graph. 

\proofcase[$|N_4| \geq 5$]\label{case:2.3:of:case:2:of:tricon24:lemma}

We first show the following claim.

\begin{claim}[resume]\label{claim:2:of:tricon:lemma}{For each~$i\in \{1,\dots,t\}$, the vertices~$u_i$ and~$u_{i+2}$ have a common neighbor other than~$v$.}

\proof We show the statement for~$i=1$.
Assuming otherwise implies that~$v$ is a singleton in~$\chi^1_{G_{(u_1,u_3)}}$. We first argue that in this coloring, the vertex~$u_2$ is also a singleton. Again, assume otherwise. Then there must be a vertex~$u\in N_4\setminus \{u_2\}$ that has the same color as~$u_2$. By Claim~\ref{claim:faces:or:diagonals}, for~$i\in \{1,2\}$ the vertices~$u_i$ and~$u_{i+1}$ have a common neighbor~$x_{i,i+1}$ so that~$u_i, v,u_{i+1},x_{i,i+1}$ form a face or a face after removing a diagonal.

Thus, the vertex~$u$ must have a neighbor~$y_{1,2}$ other than~$v$ that is adjacent to~$u_1$ and a neighbor~$y_{2,3}$ other than~$v$ that is adjacent to~$u_3$. Moreover, for~$i\in \{1,2\}$, the vertex~$x_{i,i+1}$ should have the same color as~$y_{i,i+1}$. See Figure~\ref{fig:non:equal:part}. (Note that~$y_{1,2} \neq y_{2,3}$ since we have assumed that~$u_1$ and~$u_3$ do not have a common neighbor other than~$v$.) 

\newcommand{\radiumofgraphs}{2.5}

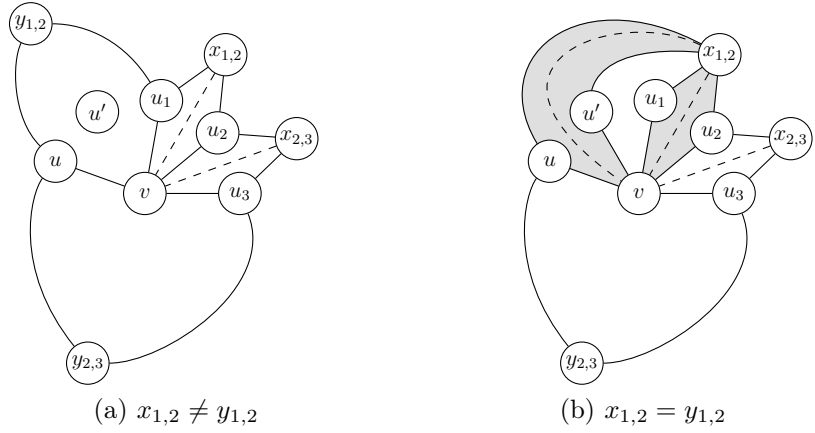
\begin{figure}
\centering
\tikzstyle{normalvertex}=[circle,fill=white,draw=black,minimum size=0.8cm,inner sep=0]
\subfloat[$x_{1,2}\neq y_{1,2}$\label{fig:non:equal:part}]{%
\begin{tikzpicture}[scale = 0.5,every node/.style={scale=0.7}]

 \node[style=normalvertex,label=right:{}] (v) at (0,0) {$v$};
 
 \node[style=normalvertex,label=right:{}] (u3) at (0:\radiumofgraphs cm) {$u_3$};
 \node[style=normalvertex,label=right:{}] (u2) at (40:\radiumofgraphs cm) {$u_2$};
 \node[style=normalvertex,label=right:{}] (u1) at (80:\radiumofgraphs cm) {$u_1$};
  
 \node[style=normalvertex,label=right:{}] (x23) at (20:1.7*\radiumofgraphs cm) {$x_{2,3}$};
 \node[style=normalvertex,label=right:{}] (x12) at (60:1.7*\radiumofgraphs cm) {$x_{1,2}$};
  
 \node[style=normalvertex,label=right:{}] (u) at (160:\radiumofgraphs cm) {$u$};
  \node[style=normalvertex,label=right:{}] (up) at (120:\radiumofgraphs cm) {$u'$};
    \node[style=normalvertex,label=right:{}] (w12) at (-3,4.5) {$y_{1,2}$};
  \node[style=normalvertex,label=right:{}] (w23) at (-1.5,-4.5) {$y_{2,3}$};
  
  \path
   (u1) edge (v) (u2) edge (v) (u3) edge (v)
   (x23) edge[dashed] (v) (x12) edge[dashed] (v)
   (x12) edge (u1) (x12) edge (u2)
   (x23) edge (u2) (x23) edge (u3)
   
   (u) edge (v)
      
  ;
\node at (-3,4.5) {};
\draw (w12) .. controls (-3.5,3.5) and (-3.5,2) .. (u);
\draw (w12) .. controls (-1.5,4.5) and (-0.5,4) .. (u1);
\draw (w23) .. controls (-3.5,-2) and (-3,0) .. (u);
\draw (w23) .. controls (0.5,-4.5) and (3.5,-2.5) .. (u3);
\end{tikzpicture}
}
\quad \quad \quad 
\tikzstyle{normalvertex}=[circle,fill=white,draw=black,minimum size=0.8cm,inner sep=0]
\subfloat[$x_{1,2} = y_{1,2}$\label{equal:part}]{%
\begin{tikzpicture}[scale = 0.5,every node/.style={scale=0.7}]

 \node[style=normalvertex,label=right:{}] (v) at (0,0) {$v$};
 
 \node[style=normalvertex,label=right:{}] (u3) at (0:\radiumofgraphs cm) {$u_3$};
 \node[style=normalvertex,label=right:{}] (u2) at (40:\radiumofgraphs cm) {$u_2$};
 \node[style=normalvertex,label=right:{}] (u1) at (80:\radiumofgraphs cm) {$u_1$};
  
 \node[style=normalvertex,label=right:{}] (x23) at (20:1.7*\radiumofgraphs cm) {$x_{2,3}$};
 \node[style=normalvertex,label=right:{}] (x12) at (60:1.7*\radiumofgraphs cm) {$x_{1,2}$};
  
 \node[style=normalvertex,label=right:{}] (u) at (160:\radiumofgraphs cm) {$u$};
 \node[style=normalvertex,label=right:{}] (up) at (120:\radiumofgraphs cm) {$u'$};
  \node[style=normalvertex,label=right:{}] (w23) at (-1.5,-4.5) {$y_{2,3}$};
  
  \path
   (u1) edge (v) (u2) edge (v) (u3) edge (v)
   (x23) edge[dashed] (v) (x12) edge[dashed] (v)
   (x12) edge (u1) (x12) edge (u2)
   (x23) edge (u2) (x23) edge (u3)
   
   (u) edge (v) (up) edge (v)
  
  ;

\begin{pgfonlayer}{edgelayer}

\filldraw[ lightgray, opacity=.5] (v.center) -- (u.center)   .. controls (-4.5,3.5) and (-1.5,6)  ..  (x12.center) .. controls (0,4) and (-1.5,3.5) .. (up.center) --(v.center)--cycle;
\filldraw[ lightgray, opacity=.5] (v.center) -- (u1.center)   -- (x12.center) -- (u2.center) --(v.center)--cycle;

\draw (x12.center) .. controls (-1.5,6) and (-4.5,3.5) .. (u.center);
\draw (x12.center) .. controls (0,4) and (-1.5,3.5) .. (up.center);
\draw (w23) .. controls (-3.5,-2) and (-3,0) .. (u);
\draw[dashed] (x12) .. controls (-0.5,5) and (-5,3.5) .. (v);
\draw (w23) .. controls (0.5,-4.5) and (3.5,-2.5) .. (u3);

\end{pgfonlayer}

\tikzset{
    vertex/.style = {
        circle,
        fill      = black,
        outer sep = 2pt,
        inner sep = 1pt,
    }
}
\end{tikzpicture}
}
\caption{An illustration of Case~\ref{case:2.3:of:case:2:of:tricon24:lemma} in the proof of Lemma~\ref{tri:con:2:individ:degree:4}. Some possible diagonals are shown as dashed lines. The highlighted regions in the drawing on the right cannot contain other vertices implying that~$\{x_{1,2},v\}$ separates~$u_1$ from~$u$.}
\label{fig:degree:4:proof:illus}
\end{figure}

Since it holds that~$|N_4| \geq 5$ we know that~$u_4\neq u_t$. Therefore~$u\neq u_{4}$ or~$u\neq u_{t}$. By symmetry we can assume the latter. (To see the symmetry recall that~$u_4$ is the successor of~$u_3$ in~$N_4$ and~$u_{t}$ is the predecessor of~$u_1$ in~$N_4$).
Note that the cycle~$v,u,y_{1,2},u_1$ separates~$u_t$ from~$u_2$.

Consider the area~$A'$ bounded by the cycle~$v,u,y_{1,2},u_1$ which contains~$u_t$. Inside the area lies~$u'\in N_4$, the vertex that follows~$u$ in the cyclic ordering of~$N_4$. 

Consider the set~$M \coloneqq N(u_3)\cup \{u_1,v,u_3\}$ and note that~$M$ is a union of color classes since we have assumed that~$u_1$ and~$u_3$ do not have a common neighbor other than~$v$. The cycle~$v,u,y_{1,2},u_1$ contains only two vertices of~$M$ and there are no vertices inside~$A'$ that are in~$M$. Thus, due to~3-connectivity, there must be a path from~$u$ or from~$y_{1,2}$ to~$u'$ such that no inner vertex of the path is in~$M$. Note that~$u'$ does not have the same color as~$u_2$ since it cannot share a neighbor other than~$v$ with~$u_3$.

Unless~$x_{1,2}= y_{1,2}$, every path from~$u_2$ and every path from~$x_{1,2}$ to a vertex in~$N_4\setminus \{u_1,u_2,u_3,u\}$ that does not have an inner vertex in~$M$  must pass through~$u$ or through~$y_{1,2}$. This implies however that~$u$ and~$u_2$ do not have the same color or that~$x_{1,2}$ and~$y_{1,2}$ do not have the same color contradicting our construction. We conclude that~$x_{1,2}=y_{1,2}$.

By Claim~\ref{claim:faces:or:diagonals}, the vertices~$u$ and~$u'$ have a common neighbor other than~$v$. If this neighbor were not~$x_{1,2}$ then~$u$ could not have the same color as~$u_2$, since there would be a path from~$u$ to~$u'$ avoiding inner vertices from~$M\cup \{y_{1,2}\}$ but there would be no such path from~$u$ to any vertex in~$N_4\setminus \{u_1,u_2,u_3,u\}$. Figure~\ref{equal:part} depicts this situation. We conclude that~$u'$ is a neighbor of~$y_{1,2}=x_{1,2}$. However this makes~$ \{y_{1,2},v\}$ a separator that separates~$u_1$ from~$u$, since by Claim~\ref{claim:faces:or:diagonals}, neither~$u'$ nor~$u_2$ have a neighbor both inside and outside of the cycle~$u',v,u_2,x_{1,2}$.

Up to this point, regarding our efforts to prove the claim, we have shown that~$u_2$ is a singleton. If~$x_{1,2}$ were not adjacent to~$v$ or~$x_{2,3}$ were not adjacent to~$v$, then~$\chi^1_{G_{(u_1,u_3)}}$ would have three singletons lying on the same face. So assume otherwise. We can also assume that neither~$x_{1,2}$ nor~$x_{2,3}$ is a singleton. But this cannot be, because the copies rendering~$x_{1,2}$ and~$x_{2,3}$ non-singletons (i.e., the other, necessarily existing vertices that have the same color as~$x_{1,2}$ or~$x_{2,3}$) should also be non-equal and adjacent to~$u_2$ which would force~$u_2$ to have degree at least~5. This proves the claim. 
\uend\end{claim} 

Since~$G[N_4]$ is empty, a common neighbor of~$u_i$ and~$u_{i+2}$ other than~$v$ must be equal to a common neighbor of~$u_{i+1}$ and~$u_{i+3}$ other than~$v$. This means that there is a vertex~$v'$ other than~$v$ adjacent to all vertices of~$N_4$.

Consider now the area bounded by the cycle~$v,u_i,v',u_{i+1}$ which does not contain~$u_{i+2}$. If both~$u_i$ and~$u_{i+1}$ have two neighbors inside this area, then they do not have any neighbors outside the area, making~$\{v,v'\}$ a separator.
If~$u_i$ only has one neighbor inside the area, then this neighbor coincides with~$x_{i,i+1}$ and hence must be adjacent to~$u_{i+1}$. We conclude that~$v,u_i,x_{i,i+1},u_{i+1}$ forms a face or becomes a face after removing the diagonal~$\{x_{i,i+1},v\}$. A symmetric argument can be applied with regard to~$v'$ in place of~$v$. It follows that inside the cycle~$v,u_i,v',u_{i+1}$ there is at most one vertex, namely~$x_{i,i+1}$.
However, we already ruled out vertices of degree~3 at the beginning of the proof, so~$x_{i,i+1}$ must be adjacent to all vertices of the cycle, and thus has degree~4. This cannot be since it would then be in~$N_4$.
We conclude that~$u_i$ does not have a neighbor inside the cycle~$v,u_i,v',u_{i+1}$. A similar observation holds for the cycle~$v,u_{i-1},v',u_i$.
 However,~$u_i$ must have some neighbors within the area bounded by the cycle~$v,u_i,v',u_{i+1}$ or within the area bounded by the cycle~$v,u_{i-1},v',u_i$ yielding the final contradiction.
 \QED
 \end{proofcases}
 \end{proofcases}
\end{proof}

\begin{proof}[Proof of Lemma~\ref{WL:dim:of:3:con:graphs}]
Recalling that every~3-connected planar graph has a vertex of degree~3,~4 or~5, the proof follows immediately by combining Lemma~\ref{lem:3:sing:on:face:implies:discrete} with the Lemmas~\ref{tri:con:2:individ:degree:5},~\ref{tri:con:2:individ:degree:3}, and~\ref{tri:con:2:individ:degree:4}.
\end{proof}

\bibliographystyle{plain}
\bibliography{planar_graphs}

 \end{document}